\documentclass[a4paper,leqno]{article}

\usepackage{natbib}
\usepackage{amsmath,amssymb}
\usepackage[dvips,final]{graphicx}
\usepackage{psfrag}
\usepackage{hyperref}

\numberwithin{equation}{section}
\newtheorem{theorem}{Theorem}

\newtheorem{lemma}[theorem]{Lemma}
\newtheorem{proposition}[theorem]{Proposition}
\newtheorem{remark}[theorem]{Remark}
\newtheorem{definition}[theorem]{Definition}
\numberwithin{theorem}{section}
\newcommand\proof{\emph{Proof. }}

\newcommand\bE{{\mathbb E}}
\newcommand\bF{{\mathbb F}}
\newcommand\bM{{\mathbb M}}
\newcommand\bP{{\mathbb P}}
\newcommand\bQ{{\mathbb Q}}
\newcommand\bR{{\mathbb R}}

\newcommand\lam{\lambda}

\newcommand\cL{{\cal L}}
\newcommand\cF{{\cal F}}
\newcommand\cB{{\cal B}}

\newcommand\cQ{{\cal Q}}

\newcommand\cT{{\cal T}}
\newcommand\cM{{\cal M}}
\newcommand\cW{{\cal W}}

\newcommand\Ltwo{{\mathrm{L}^2}}
\newcommand\Ltwoloc{{\mathrm{L}^2_{\mathrm{loc}}}}
\newcommand\p{^\prime}
\newcommand\pp{^{\prime\prime}}
\newcommand\sig{\sigma}

\newcommand\bstar{\begin{eqnarray*}}
\newcommand\estar{\end{eqnarray*}}
\newcommand\be{\begin{equation}}
\newcommand\ee{\end{equation}}
\newcommand\bea{\begin{eqnarray}}
\newcommand\eea{\end{eqnarray}}
\newcommand\1{{\bf 1}}
\renewcommand\ni{\noindent}
\newcommand\clc{,\ldots ,}

\newcommand\prv{P^{\text{\tiny RV}(w)}_T}

\newcommand{\R}{\mathbb{R}}
\newcommand{\id}{\mathrm{d}}

\newcommand{\func}{\lambda}
\newcommand{\nmax}{\bar{n}}
\newcommand{\nmin}{\underline{n}}

\newcommand\ba{\mathbf{a}}
\newcommand\bb{\mathbf{b}}
\newcommand{\bx}{\mathbf{x}}
\newcommand{\by}{\mathbf{y}}
\newcommand{\cy}{\check{\by}}
\newcommand{\bK}{\mathbf{K}}

\renewcommand{\t}{\tilde}
\newcommand\tr{^{\mbox{\tiny\textit{T}}}}
\newcommand\da{^\dagger}
\newcommand\pathspc{\mathcal{P}}

\usepackage{color}

\title{\textbf{\Large Arbitrage bounds for prices of weighted variance swaps}
 \footnote{This paper was previously circulated under the title ``Arbitrage Bounds for Prices of Options on Realized Variance."}
}

\author{{\sc Mark Davis}\thanks{Department of Mathematics, Imperial College London, London SW72AZ, UK ({\tt mark.davis@imperial.ac.uk}).}
        \and {\sc Jan Ob\l\'oj}\thanks{Mathematical Institute, Oxford-Man Institute of Quantitative Finance and St John's College, University of Oxford, Oxford OX1 3LB, UK ({\tt obloj@maths.ox.ac.uk})}
 \and {\sc Vimal Raval}\thanks{Imperial College London. Work supported by EPSRC under a Doctoral Training Award.}}

\date{}

\begin{document}
\maketitle

\begin{abstract}
We develop a theory of robust pricing and hedging of a weighted variance swap given market prices for a finite number of co--maturing put options. We assume the put option prices do not admit arbitrage and deduce no-arbitrage bounds on the weighted variance swap along with super- and sub- replicating strategies that enforce them. We find that market quotes for variance swaps are surprisingly close to the model-free lower bounds we determine. We solve the problem by transforming it into an analogous question for a European option with a convex payoff. The lower bound becomes a problem in semi-infinite linear programming which we solve in detail. The upper bound is explicit. 

We work in a model-independent and probability-free setup. In particular we use and extend F\"ollmer's pathwise stochastic calculus. Appropriate notions of arbitrage and  admissibility are introduced. This allows us to establish the usual hedging relation between the variance swap and the `log contract' and similar connections for weighted variance swaps. Our results take form of a FTAP: we show that the absence of (weak) arbitrage is equivalent to the existence of a classical model which reproduces the observed prices via risk--neutral expectations of discounted payoffs.


\end{abstract}



\pagestyle{myheadings}
\thispagestyle{plain}
\markboth{M Davis, J. Ob\l\'oj and V. Raval}{Arbitrage bounds for prices of options on realised variance}

\section{Introduction}\label{sec:intro}
In the practice of quantitative finance, the risks of `model error' are by now universally appreciated. Different models, each perfectly calibrated to market prices of a set of liquidly traded instruments, may give widely different prices for contracts outside the calibration set. The implied hedging strategies, even for contracts within the calibration set, can vary from accurate to useless, depending on how well the model captures the sample path behaviour of the hedging instruments.
 This fundamental problem has motivated a large body of research ranging from asset pricing through portfolio optimisation, risk management to macroeconomics (see for example\ \cite{Cont:06}, \cite{FollmerSchiedWeber:09}, \cite{AcciaioFollmerPenner:11}, \cite{HansenSargent:10} and the references therein).  Another stream of literature, to which this paper is a contribution, develops a robust approach to mathematical finance, see \cite{hobson_lookback}, \cite{dh05}, \cite{CoxObloj:11}. In contrast to the classical approach, no probabilistic setup is assumed. Instead we suppose we are given current market quotes for the underlying assets and some liquidly traded options. We are interested in no-arbitrage bounds on a price of an option implied by this market information. Further, we want to understand if, and how, such bounds may be enforced, in a model-independent way, through hedging. 

In the present paper we consider robust pricing and hedging of a weighted variance swap when prices of a finite number of co--maturing put options are given. 
If $(S_t, t\in[0,T])$ denotes the price of a financial asset, the \emph{weighted realized variance} is defined as
\be 
RV_T=\sum_{i=1}^nh(S_{t_i})\left(\log\frac{S_{t_i}}{S_{t_{i-1}}}\right)^2,\label{rv}
\ee
where $t_i$ is a pre-specified sequence of times $0=t_0<t_1<\cdots<t_n=T$, in practice often daily sampling, and $h$ is a given weight function. A \emph{weighted variance swap} is a forward contract in which cash amounts equal to $A\times RV_T$ and $A\times P^{\text{\tiny RV}}_T$ are exchanged at time $T$, where $A$ is the dollar value of one variance point and $P^{\text{\tiny RV}}_T$ is the variance swap `price', agreed at time 0. Three representative cases considered in this paper are the \emph{plain vanilla variance swap} $h(s)\equiv 1$, the \emph{corridor variance swap} $h(s)=\1_I(s)$ where $I$ is a possibly semi-infinite interval in $\bR^+$, and the \emph{gamma swap} in which $h(s)=s$. The reader can consult  \citet{gatheral} for information about variance swaps and more extended discussion of the basic facts presented below. In this section we restrict the discussion to the vanilla variance swap; we return to the other contracts in Section \ref{sec-wvs}.

In the classical approach, if we model $S_t$ under a risk-neutral measure $\bQ$ as\footnote{For example, $S_t=(e^{(r-q)t}S_0)(e^{\sig W_t-\frac{1}{2}\sig^2t})$ in the Black-Scholes model, with the conventional notation.} $S_t=F_t\exp(X_t-\frac{1}{2}\langle X\rangle_t)$, where $F_t$ is the forward price (assumed to be continuous and of bounded variation) and $X_t$ is a continuous martingale with quadratic variation process $\langle X\rangle_t$, then $S_t$ is a continuous semimartingale and
\[\log\frac{S_{t_i}}{S_{t_{i-1}}}=(g(t_i)-g(t_{i-1}))+(X_{t_i}-X_{t_{i-1}}),\]
with $g(t)=\log(F_{t_i})-\frac{1}{2}\langle X\rangle_{t_i}$.
For $m=1,2\ldots$ let $\{s^m_i, i=0\clc k_m\}$ be the ordered set of stopping times in $[0,T]$ containing the times $t_i$ together with times 
$\tau^m_0=0,\tau_k^m=\inf\{t>\tau^m_{k-1}:|X_t-X_{\tau_{k-1}^m}|>2^{-m}\}$ wherever these are smaller than $T$. If $RV^m_T$ denotes the realized variance computed as in \eqref{rv} but using the times $s^m_i$, then $RV^m_T\to\langle X\rangle_T=\langle\log S\rangle_T$ almost surely, see \citet{rogwil}, Theorem IV.30.1. For this reason, most of the pricing literature on realized variance studies the continuous-time limit $\langle X\rangle_T$ rather than the finite sum \eqref{rv} and we continue the tradition in this paper.

The key insight into the analysis of variance derivatives in the continuous limit was provided by \citet{neuberger_94}. 
We outline the arguments here, see Section \ref{sec-wvs} for all the details. With $S_t$ as above and $f$ a $C^2$ function, it follows from the general It\^o formula that $Y_t=f(S_t)$ is a continuous semimartingale with $\id\langle Y\rangle_t=(f'(S_t))^2\id\langle S\rangle_t$. In particular, with the above notation
\bstar \id(\log S_t)&=& \frac{1}{S_t}\id S_t-\frac{1}{2S_t^2}\id\langle S\rangle_t
= \frac{1}{S_t}\id S_t -\frac{1}{2}\id\langle \log S\rangle_t,\estar
so that
\be \langle \log S\rangle_T=2\int_0^T\frac{1}{S_t}\id S_t-2\log(S_T/S_0)
\label{log}\ee
This shows that the realized variation is replicated by a portfolio consisting of a self-financing trading strategy that invests a constant \$2$D_T$ in the underlying asset $S$ together with a European option whose exercise value at time $T$ is $\lam(S_T)=-2\log(S_T/F_T)$. Here $D_T$ is the time-$T$ discount factor.
Assuming that the stochastic integral is a martingale, \eqref{log} shows that the risk-neutral value of the variance swap rate $P^{\text{\tiny RV}}_T$ is
\be P^{\text{\tiny RV}}_T  =\bE\Big[\langle\log S\rangle_T\Big]=-2\bE[\log(S_T/F_T)]=2\frac{1}{D_T}P_{\log},\label{ptiny}\ee 
that is, the variance swap rate is equal to the forward  value $P_{\log}/D_T$  of two `log contracts'---European options with exercise value $-\log(S_T/F_T)$, a convex function of $S_T$. \eqref{ptiny} gives us a way to evaluate $P^{\text{\tiny RV}}_T$ in any given model. The next step is to rephrase $P_{\log}$ in terms of call and put prices only.

Recall that for a convex function $f:\bR^+\to\bR$ the recipe $f''(a,b]=f'_+(b)-f'_+(a)$ defines a positive measure $f''(\id x)$ on $\cB(\bR^+)$, equal to $f''(x)\id x$ if $f$ is $C^2$. We then have the Taylor formula
\be f(x)=f(x_0)+f'_+(x_0)(x-x_0)+\int_0^{x_0}(y-x)^+\nu(\id y)+\int_{x_0}^\infty(x-y)^+f''(\id y).\label{tf}\ee
 Applying this formula with $x=S_T, x_0=F_T$ and $f(s)=\log(s/S_0)$, and combining with \eqref{log} gives
\be \frac{1}{2}\langle\log S\rangle_T=\int_0^T\frac{1}{S_t}\id S_t -\log(F_T/S_0) +1-S_T/F_T+\int_0^{F_T}\frac{(K-S_T)^+}{K^2}\id K
+\int_{F_T}^\infty\frac{(S_T-K)^+}{K^2}\id K.\label{ph}\ee
Assuming that puts and calls are available for all strikes $K\in\bR^+$ at traded prices $P_K,C_K$, \eqref{ph} provides a perfect hedge for the realized variance in terms of self-financing dynamic trading in the underlying (the first four terms) and a static portfolio of puts and calls and hence uniquely specifies the variance swap rate
$P^{\text{\tiny RV}}_T $ as
 \be P^{\text{\tiny RV}}_T =\frac{2}{D_T}\int_0^\infty\frac{1}{K^2}(P_K\1_{(K\leq F_T)}+C_K\1_{(K>F_T)})\id K.\label{prv_val}\ee

Our objective in this paper is to investigate the situation where, as in reality, puts and calls are available at only a finite number of strikes. In this case we cannot expect to get a unique value for $ P^{\text{\tiny RV}}_T$ as in \eqref{prv_val}. However, from \eqref{log} we expect to  obtain arbitrage bounds on $P^{\text{\tiny RV}}_T$ from bounds on the value of the log contract---a European option whose exercise value is the convex function $-\log(S_T/F_T)$. It will turn out that the weighted variance swaps we consider are associated in a similar way with other convex functions. In  \citet{dh05} conditions were stated under which a given set of prices $(P_1\clc P_n)$ for put options with strikes $(K_1\clc K_n)$ all maturing at the same time $T$ is consistent with absence of arbitrage.
Thus the first essential question we have to answer is: given a set of put prices consistent with absence of arbitrage and a convex function $\lam_T$, what is the range of prices $P_\lam$ for a European option with exercise value $\lam_T(S_T)$ such that this option does not introduce arbitrage when traded in addition to the existing puts? We answer this question in Section \ref{sec-cp} following a statement of the standing assumptions and a summary of the \cite{dh05} results in Section \ref{sec:probform}. 

The basic technique in Section \ref{sec-cp} is to apply the duality theory of semi-infinite linear programming to obtain the values of the most expensive sub-replicating portfolio and the cheapest super-replicating portfolio, where the portfolios in question contain static positions in the put options and dynamic trading in the underlying asset. The results for the lower and upper bounds are given in Propositions \ref{lb} and \ref{ub2} respectively. (Often, in particular in the case of plain vanilla variance swaps, the upper bound is infinite.) For completeness, we state and prove the fundamental Karlin-Isii duality theorem in an appendix, Appendix \ref{sec-app}.
With these results in hand, the arbitrage conditions are stated and proved in Theorem \ref{convex_payoff}.

In this first part of the paper---Sections \ref{sec:probform} and \ref{sec-cp}---we are concerned exclusively with European options whose exercise value depends only on the asset price $S_T$ at the common exercise time $T$. In this case the sample path properties of $\{S_t, t\in[0,T]\}$ play no role and the only relevant feature of a `model' is the marginal distribution of $S_T$. For this reason we make no assumptions about the sample paths. When we come to the second part of the paper, Sections \ref{sec-wvs} and \ref{comp}, analysing weighted variance swaps, then of course the sample path properties are of fundamental importance. Our minimal assumption is that $\{S_t, t\in[0,T]\}$
is continuous in $t$. However, this is not enough to make sense out of the problem. The connection between the variance swap and the `log contract', given at \eqref{log}, is based on stochastic calculus, assuming that $S_t$ is a continuous semimartingale. In our case the starting point is simply a set of prices. No model is provided, but we have to define what we mean by the continuous-time limit of the realized variance and by continuous time trading. An appropriate framework is the pathwise stochastic calculus of \citet{Fol81} where the quadratic variation and an It\^o-like integral are defined for a certain subset $\cQ$ of the space of continuous functions $C[0,T]$. Then \eqref{log} holds for $S(\cdot)\in\cQ$ and we have the connection we sought between variance swap and log contract. For weighted variance swaps the same connection exists, replacing `$-\log$' by some other convex function $\lam_T$, as long as the latter is $C^2$. In the case of the corridor variance swap, however, $\lam''_T$ is discontinuous and we need to restrict ourselves to a smaller class of paths $\cL$ for which the It\^o formula is valid for functions in the Sobolev space $\cW_2$. All of this is consistent with models in which $S_t$ is a continuous semimartingale, since then $\bP\{\omega:S(\cdot,\omega)\in\cL\}=1$. We discuss these matters in a separate appendix, Appendix \ref{sec-pathwise}.

With these preliminaries in hand we formulate and prove, in Section \ref{sec-wvs} the main result of the paper, Theorem \ref{thm:wVS_arb}, giving conditions under which the quoted price of a weighted variance swap is consistent with the prices of existing calls and puts in the market, assuming the latter are arbitrage-free. When the conditions fail there is a weak arbitrage, and we identify strategies that realize it.

In mathematical finance it is generally the case that option bounds based on super- and sub-replication are of limited value because the gap is too wide. However, here we know that as the number of put options increases then in the limit there is only one price, namely, for a vanilla variance swap, the number $P^{\text{\tiny RV}}_T$ of \eqref{prv_val}, so we may expect that in a market containing a reasonable number of liquidly-traded puts the bounds will be tight. In Section \ref{comp} we present data from the S\&P500 index options market which shows that our computed lower bounds are surprisingly close to the quoted variance swap prices.

This paper leaves many related questions unanswered. In particular we would like to know what happens when there is more than one exercise time $T$, and how the results are affected if we allow jumps in the price path. The latter question is considered in the small time-to-expiry limit in the recent paper  \citet{KRMK10}. A paper more in the spirit of ours, but taking a complementary approach, is \citet{HobsonKlimmek:11}[HK] . In our paper the variance swap payoff is defined by the continuous-time limit, and we assume that prices of a finite number of traded put and/or call prices are known. By contrast, HK deal with the real variance swap contract (i.e., discrete sampling) but assume that put and call prices are known for all strikes $K\in\bR^+$, so in practice the results will depend on how the traded prices are interpolated and extrapolated to create the whole volatility surface. Importantly, HK allow for jumps in the price process; this is indeed the main focus of their work. The price bounds they obtain are sharp in the continuous-sampling limit. Combining our methods with those of HK would be an interesting, and possibly quite challenging, direction for future research, see Remark \ref{rk:Hobson_Klimmek} below.

Exact pricing formulas for a wide variety of contracts on discretely-sampled realized variance, in a general L\'evy process setting, are provided in \cite{CroDav11}. 
\section{Problem formulation}\label{sec:probform}
Let $(S_t, t\in[0,T])$ be the price of a traded financial asset, which is assumed non-negative: $S_t\in\bR^+=[0,\infty)$. In addition to this underlying asset, various derivative securities as detailed below are also traded. All of these derivatives are European contracts maturing at the same time $T$. The present time is 0. We make the following standing assumptions as in  \citet{dh05}:

\ni(i) The market is frictionless: assets can be traded in arbitrary amounts, short or long, without transaction costs, and the interest rates for borrowing and lending are the same.

\ni(ii) There is no interest rate volatility. We denote by $D_t$ the market discount factor for time $t$, i.e. the price at time $0$ of a zero-coupon bond maturing at $t$.

\ni(iii) There is a uniquely-defined forward price $F_t$ for delivery of one unit of the asset at time $t$.  This will be the case if the asset pays no dividends or if, for example, it has a deterministic dividend yield. We let $\Gamma_t$ denote the number of shares that will be owned at time $t$ if dividend income is re-invested in shares; then $F_0=S_0$ and $F_t=S_0/(D_t\Gamma_t)$ for $t>0$.

We suppose that $n$ put options with strikes $0<K_1<\ldots<K_n$ and maturity $T$ are traded, at time-0 prices $P_1\clc P_n$. In \citep{dh05}, the arbitrage relations among these prices were investigated. The facts are as follows. A \emph{static portfolio} $X$ is simply a linear combination of the traded assets, with weights $\pi_1\clc \pi_n,\phi,\psi$ on the options, on the underlying asset and on cash respectively, it being assumed that dividend income is re-invested in shares. The value of the portfolio at maturity is
\be X_T=\sum_{i=1}^n\pi_i[K_i-S_T]^++\phi \Gamma_T S_T+\psi D_T^{-1}\label{XT}\ee
and the set-up cost at time 0 is
\be X_0=\sum_{i=1}^n\pi_iP_i+\phi S_0+\psi.\label{X0}\ee
Note that it does not make any sense a priori to speak about $X_t$ -- the value of the portfolio at any intermediate time -- as this is not determined by the market input and we do not assume the options are quoted in the market at intermediate dates $t\in (0,T)$. To make sense to statements like ``$X_T$ is non-negative" we need to specify the universe of scenarios we consider. Later, in Section \ref{sec-wvs}, this will mean the space of of possible paths $(S_t:t\leq T)$. However here, since we only allow static trading as in \eqref{XT}, we essentially look at a one--period model and we just need to specify the possibly range of values for $S_T$ and we suppose that $S_T$ may take any value in $[0,\infty)$. 

A \emph{model} $\cM$ is a filtered probability space $(\Omega,\bF=(\cF_t)_{t\in[0,T]},\bP)$  together with a positive $\bF$-adapted process $(S_t)_{t\in[0,T]}$ such that $S_0$ coincides with the given time-0 asset price. Given market prices of options, we say that $\cM$  is a \emph{market model} for these options if $M_t=S_t/F_t$ is an $\bF$-martingale (in particular, $S_t$ is integrable) and the market prices equal to the $\bP$--expectations of discounted payoffs. In particular, $\cM$ is a market model for put options if $P_i=D_T \bE[K_i-S_T]^+$ for $i=1\clc n$. 
We simply say that $\cM$ is a market model if the set of market options with given prices is clear from the context.
It follows that in a market model we have joint dynamics on $[0,T]$ of all assets' prices such that the initial prices agree with the market quotes and the discounted prices are martingales. By the (easy part of the) First Fundamental Theorem of Asset Pricing \citep{DelbaenSchachermayer:04} the dynamics do not admit arbitrage. 

Our main interest in the existence of a market model and we want to characterise it in terms of the given market quoted prices. This requires notions of arbitrage in absence of a model. We say that there is \emph{model-independent arbitrage} if 
one can form a portfolio with a negative setup cost and a non-negative payoff. There is \emph{weak arbitrage} if for any model there is a static portfolio $(\pi_1\clc \pi_n,\phi,\psi)$ such that $X_0\leq 0$ but $\bP(X_T\geq 0)=1$ and $\bP[X_T>0]>0$. In particular, a model independent arbitrage is a special case of weak arbitrage, as in \citet{co09dnt}. Throughout the paper, we say that prices are \emph{consistent with absence of arbitrage} when they do not admit a weak arbitrage.

It is convenient at this point to move to normalized units. If for $i=1\clc n$ we define
\be p_i=\frac{P_i}{D_T F_T },\qquad k_i=\frac{K_i}{F_T },\label{nun}\ee
then, in a market model, we have $P_i=D_T F_T \bE[K_i/F_T -M_T ]^+$ and hence
\[ p_i=\bE[k_i-M_T ]^+.\]

Also, we can harmlessly introduce another put option with strike $k_0=0$ and price $p_0=0$. Let $\underline{n}=\max\{i:p_i=0\}$,\,$\overline{n}=\inf\{i:p_i=k_i-1\}(=+\infty \mbox{ if } \{\cdots\}=\emptyset)$ and $\bK=[k_{\nmin},k_{\nmax}]$, where it is hereafter understood as $\bK=[k_{\nmin},\infty)$ if $\nmax=\infty$.

The main result of \citet[Thm.~3.1]{dh05}, which we now recall, says that there is no weak arbitrage if and only if there exists a market model. The conditions are illustrated in Figure \ref{puts}.

\begin{proposition}\label{prop:DH}
Let $r:[0,k_n]\to\bR^+$ be the linear interpolant of the points $(k_i,p_i), i=0\clc n$. Then there exists a market model if and only if $r(\cdot)$ is a non-negative, convex, increasing function such that $r(0)=0, r(k)\geq [k-1]^+$ and $r'_-(k_{\overline{n}\wedge n})<1$, where $r'_-$ denotes the left-hand derivative. If these conditions hold except that $\nmax=\infty$ and $r'_-(k_n)=1$ then there is weak arbitrage. Otherwise there is a model-independent arbitrage. Furthermore, if a market model exists then one may choose it so that the distribution of $S_T$ has finite support.
\end{proposition}
\begin{figure}[t]
\begin{center}
\includegraphics[scale=0.5] {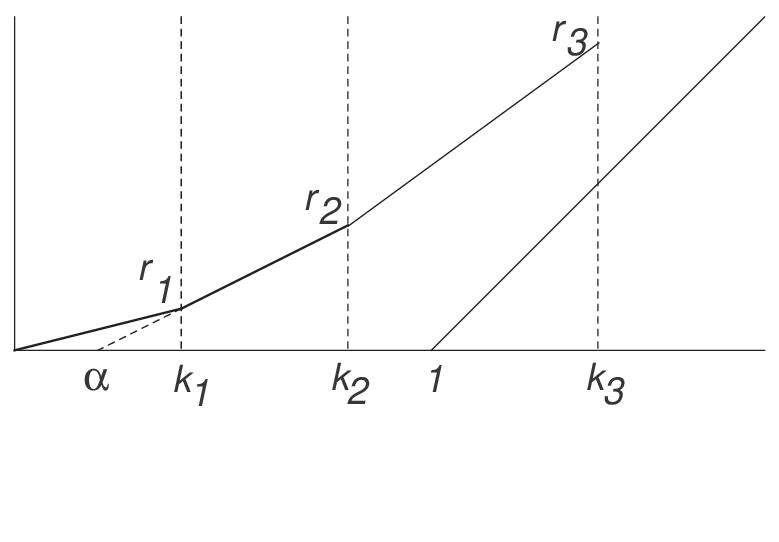}
\vspace{-12mm}\caption{\label{puts} Normalized put prices $(k_i,p_i)$ consistent with absence of arbitrage. An additional put with price $r=0$ and strike $k\in[0,\alpha]$ does not introduce arbitrage.}
\end{center}
\end{figure}

\textsc{Remarks.} (i) In \citet{dh05}, the case of call options was studied. The result stated here follows by put-call parity, valid in view of our frictionless markets assumption. The normalised call price is $c_i=p_i+1-k_i$.

(ii) Of course, no put option would be quoted at zero price, so in applications $\underline{n}=0$ always. As will be seen below, it is useful for analysis to include the artificial case $\underline{n}>0$.

Throughout the rest of the paper, we assume that \emph{the put option prices $(P_1\clc P_n)$ do not admit a weak arbitrage and hence there exists a market model consistent with the put prices}.   Let $\cM$ be a market model and let $\mu$ be the distribution of $M_T$ in this model. Then $\mu$ satisfies

\begin{subequations}\label{mucond}
\begin{align}
 \int_{\bR^+}1\,\mu(\id x)&=1\label{mucond_a}\\
 \int_{\bR^+}x\,\mu(\id x)&=1\label{mucond_b}\\
 \int_{\bR^+}[k_i-x]^+\mu(\id x)&=p_i,\quad i=1\clc n.\label{mucond_c}
\end{align}
\end{subequations}

 Conversely, given a probability measure $\mu$ on $\bR^+$ which satisfies the above we can construct a market model $\cM$ such that $\mu$ is the distribution of $M_T$. For example, let $(\Omega,\bF,\bP, (W_t)_{t\in\bR^+})$ be the Wiener space. By the Skorokhod embedding theorem (cf. \citep{genealogia}), there is a stopping time $\tau$ such that $W_\tau\sim\mu$ and $(W_{\tau\land t})$ is a uniformly integrable martingale. It follows that we can put $M_t=1+W_{\tau\wedge (t/(T-t))}$ for $t\in[0,T)$. This argument shows that the search for a market model reduces to a search for a measure $\mu$ satisfying \eqref{mucond}. 
We will denote by $\bM_P$ the set of measures $\mu$ satisfying the conditions \eqref{mucond}.

\begin{lemma}\label{support} For any $\mu\in \bM_P$, $\mu(\bR^+\backslash \bK)=0$.
\end{lemma}
\proof That $\mu[0,k_{\nmin})=0$ when $\nmin>0$ follows from \eqref{mucond_c} with $i=\nmin$. When $\nmax\leq n$ we have $c_{\nmax}=0$, i.e. there is a free call option with strike $k_{\nmax}$ and we conclude that $\mu(k_{\nmax},\infty)=0$.\hfill$\square$

The question we wish to address is whether, when prices of additional options are quoted, consistency with absence of arbitrage is maintained. As discussed in Section \ref{sec:intro}, we start by considering the case where one extra option is included, a European option maturing at $T$ with convex payoff.

\section{Hedging convex payoffs}\label{sec-cp}
Suppose that, in addition to the $n$ put options, a European option is offered at price $P_\lam$ at time 0, with exercise value $\lam_T(S_T)$ at time $T$, where $\lam_T$ is a convex function. We can obtain lower and upper bounds on the price of $\lam_T$ by constructing sub-replicating and super-replicating static portfolios in the other traded assets. These bounds are given in Sections \ref{sec-lower} and \ref{sec-upper} respectively and are combined in Section \ref{sec-arb} to obtain the arbitrage conditions on the price $P_\lam$.

We work in normalised units throughout, that is, the static portfolios have time-$T$ values that are linear combinations of cash, underlying $M_T$ and option exercise values $[k_i-M_T]^+$. The prices of units of these components at time 0 are $D_T, \,D_T$ and $D_Tp_i$ respectively, where a unit of cash is \$1. Indeed, to price $M_T$ observe that \$1 invested in the underlying at time 0 yields $\Gamma_TS_T/S_0=S_T/F_TD_T=M_T/D_T$ at time $T$.
To achieve a consistent normalization for $\lam_T$ we define the convex function $\lam$ as 
\begin{equation}\label{eq:def_lam}
\lam(x)=\frac{1}{F_T}\lam_T(F_Tx).
\end{equation}
In a market model $\cM$ we have $P_\lam=D_T\bE[\lam_T(S_T)]=D_TF_T\bE[\lam(M_T)]$, so the normalized price is
\[ p_\lam=\frac{P_\lam}{D_TF_T}=\bE[\lam(M_T)]\]
and the cost for delivering a payoff $\lam(M_T)$ is $D_Tp_\lam$.

\subsection{Lower bound}\label{sec-lower} 
A sub-replicating portfolio  is a static portfolio formed at time 0 such that its value at time $T$ is majorized by $\lam(M_T)$ for all values of $M_T$. Obviously, a necessary condition for absence of model independent arbitrage is that $D_Tp_\lam$ be not less than the set-up cost of any sub-replicating portfolio. It turns out that the options $k_i$ with $i\leq\nmin$ or $i\geq\nmax$ are redundant, so the assets in the portfolio are indexed by $k=1\clc m$ where
\[ m=(n+1)\wedge\nmax-\nmin+1\]
and the time-$T$ values of these assets, as functions of $x=M_T$ are
\bea a_1(x)&=&1\quad\mbox{(Cash)}\nonumber\\
a_2(x)&=& x\quad\mbox{(Underlying)}\label{exvals}\\
a_{i+2}(x)&=&[k_{\nmin+i}-x]^+,\quad i=1\clc m-2\quad\mbox{(Options)}.\nonumber
\eea
We let $\ba(x)$ be the $m$-vector with components $a_k(x)$.
Note that $a_m(x)$ is equal to $[k_{\nmax-1}-x]^+$ if $\nmax\leq n$ and to $[k_n-x]^+$ otherwise. The set-up costs for the components in \eqref{exvals}, as observed above, are $D_T,D_T$ and $D_T p_{\nmin+i}$ respectively. The corresponding \emph{forward} prices are as in \eqref{mucond}:
\bea b_1&=&1\nonumber\\
b_2&=& 1\\
b_{i+2}&=& p_{\nmin+i},\quad i=1\clc m-2.\nonumber
\label{setup}\eea
We let $\bb$ denote the $m$-vector of the forward prices. A static portfolio is defined by a vector $\by$ whose $k$th component is the number of units of the $k$th asset in the portfolio. The forward set-up cost is $\by\tr \bb$ and the value at $T$ is $\by\tr \ba(M_T)$.

With this notation, the problem of determining the most expensive sub-replicating portfolio is equivalent to solving the (primal) semi-infinite linear program
\[ P_\mathrm{LB}:\quad \sup_{\by\in\bR^m}\by\tr \bb\quad\mathrm{ subject\, to}\quad\by\tr \ba(x)\leq\lam(x)\,\,\forall x\in\bK.\]
The constraints are enforced only for $x\in\bK$. If $\nmin>0$ [$\nmax\leq n$] we have a free put with strike $k_{\nmin}$ [call with strike  $k_{\nmax}$] and, since $\lam$ is convex, we can extend the sub-replicating portfolio to all of $\bR^+$ at no cost. 

The key result here is the basic duality theorem of semi-infinite linear programming, due to \citet{Isii_1960} and Karlin, see \citet{karlin_studden}. This theorem, stated as Theorem \ref{chconvthm:isii}, and its proof are given in Appendix \ref{sec-app}. The dual program corresponding to $P_{\mathrm{LB}}$ is
\[ D_\mathrm{LB}:\quad \inf_{\mu\in\bM}\int_\bK\lam(x)\mu(\id x)\quad\mathrm{ subject\, to}\quad \int_\bK\ba(x)\mu(\id x)=\bb^0,     \]
where $\bM$ is the set of Borel measures such that each $a_i$ is integrable. The constraints in $D_\mathrm{LB}$ can be expressed as $\mu$ satisfying \eqref{mucond} for $\nmin<i<\nmax$. This is simply equivalent to $\mu\in \bM_P$ since, as shown in Lemma \ref{support}, any $\mu\in\bM_P$ has support in $\bK$. Let $V^L_P$ and $V^L_D$ be the values of the primal and dual problems respectively. It is a general and easily proved fact that $V^L_P\leq V^L_D$. The `duality gap' is $V^L_D-V^L_P$. The Karlin-Isii theorem gives conditions under which there is no duality gap and we have existence in $P_{\mathrm{LB}}$.
\begin{proposition}\label{lb}
We suppose as above that $\lam(x)$ is a convex function on $\bR^+$, finite for all $x>0$, and that $(k_i,p_i)$ is a set of normalised put option strike and price pairs which do not admit a weak arbitrage. If $\lam(x)$ is unbounded as $x\to 0$ and $\nmin=0$ then we further assume that $p_1/k_1<p_2/k_2$. Then 
$V^L_D=V^L_P$ and there exists a maximising vector $\hat{\by}$. The most expensive sub-replicating portfolio of a European option with payoff $\lam_T(S_T)$ at maturity $T$ is the static portfolio $X\da$ as in \eqref{XT} with weights $\psi\da=F_TD_T\hat{\by}_1, \phi\da=\hat{\by}_2/\Gamma_T$, $\pi_{\nmin+i}\da=\hat{\by}_{2+i}$ for $i=1\clc m-2$ and $\pi_i\da=0$ otherwise. For this portfolio, $X\da_0=D_TF_TV^L_D$.

If there is existence in the dual problem $D_{\mathrm{LB}}$ then there is an optimal measure $\mu\da$ which is a finite linear combination of Dirac measures $\mu\da=\sum_{j=1}^mw_j\delta_{x_j}(\id x)$ such that each interval $[k_j,k_{j+1})$ contains at most one point $x_j$. For this measure
\be 
\mu\da(\{x\})>0\quad\Rightarrow\quad \left.{X\da_T}\right|_{S_T=F_Tx}= \sum_{i=1}^n\pi\da_i[K_i-F_Tx]^++\phi\da\Gamma_T F_Tx+\psi\da D_T^{-1}=\lam_T(F_Tx).
\label{impl}\ee
\end{proposition}
\proof The first part of the proposition is an application of Theorem \ref{chconvthm:isii}. The primal problem $P_\mathrm{LB}$ is feasible because any support line corresponds to a portfolio (containing no options). The functions $a_1\clc a_m$ are linearly independent. Recall from Proposition \ref{prop:DH} that if $\{(k_i,p_i)\}$ do not admit weak arbitrage then there is a measure $\mu$ satisfying the conditions \eqref{mucond} and such that $\mu$ is a finite weighted sum of Dirac measures. It follows that $\int_{\bR^+}|\lam(x)|\mu(\id x)<\infty$ unless one of the Dirac measures is placed at $x=0$ and $\lam$ is unbounded at zero. If $\nmin>0$ then there is no mass on the interval $[0,\nmin)$, hence none at zero. When $\nmin=0$ we always have $p_1/k_1\leq p_2/k_2$. If $p_1/k_1=p_2/k_2$ then the payoff $[k_2-M_T]^+-k_2[k_1-M_T]^+/k_1$ has null cost and is strictly positive on $(0,k_2)$. Since $p_1>0$ there must be some mass to the left of $k_1$, and this mass must be placed at 0, else there is an arbitrage opportunity. But then $\int\lam d\mu=+\infty$ and $V^L_D=+\infty$. The condition in the proposition excludes this case. In every other case there is a realizing measure $\mu$ such that $\mu(\{0\})=0$. 
Indeed, if $p_1/k_1<p_2/k_2$ then the extended set of put prices $\{(k,0),(k_1,p_1)\clc(k_n,p_n)\}$ is consistent with absence of arbitrage if $k\in[0,\alpha]$, where $\alpha=(k_1p_2-k_2p_1)/(p_2-p_1)$ (see Figure \ref{puts}). Any model realizing these prices puts weight 0 on the interval $[0,k)$. Thus $V^L_D$ is finite under the conditions we have stated. It remains to verify that the vector $\bb$ belongs to the interior of the moment cone $M_m$ defined at \eqref{mc}. For this, it suffices to note that for all $i$ such that $k_i\in(k_{\nmin},k_{\nmax})$ it holds that $[k_i-1]^+<p_i<k_i$, and so the condition is satisfied. We now conclude from Theorem \ref{chconvthm:isii} that $V^L_P=V^L_D$ and that we have existence in the primal problem. The expressions for $\psi\da$ etc. follow from the relationships \eqref{nun} between normalized and un-normalized prices.

Assume now that the dual problem has a solution. Any optimal measure $\mu^{\dag}\in \bM_P$ satisfies
$$ \int_{\bK} \func(x) \mu^{\dag}(\id x) = \inf_{\mu \in \bM_P} \left\{
  \int_{\bK}\func(x) \mu(\id x) \right\}.$$
Recall $\bK = [k_{\nmin}, k_{\nmax}]$ and partition $\bK$ into intervals $I_{\nmin +1}, \ldots, I_{\nmax \wedge n +1}$ defined by
$$ I_i = [k_{i-1}, k_i)\ \textrm{ for $i = \nmin +1, \ldots \nmax
      \wedge n$ and }
  I_{\nmax \wedge n + 1} = [k_{\nmax \wedge n}, k_{\nmax}),$$
so $I_{\nmax \wedge n + 1} = \emptyset$ if $\nmax \leq n$.  Lemma \ref{l1} below asserts that we may take $\mu\da$ atomic with at most one atom in each of the intervals $I_i$. By definition the optimal subhedging portfolio $X\da$ satisfies $X\da_T\leq\lam_T(S_T)$ while our duality result shows that the $\mu\da$ expectations of these random variables coincide. Hence $X\da_T=\lam_T(S_T)$ a.s.\ for $M_T=S_T/F_T$ distributed according to $\mu\da$, and \eqref{impl} follows. \hfill$\square$
\begin{lemma}\label{l1}
Let $\mu
\in \bM_P$ and suppose $\int_{\bK} |\func(x)| \mu( \id x) < \infty$. Define
\be \mathcal{I}_{\mu} = \{i \leq \nmax \wedge n+1|
    \mu(I_i)>0\}. \label{Iset} \ee
Now let $\mu'$ be the measure
$$ \mu' = \sum_{i \in \mathcal{I}_{\mu}} \mu(I_i) \delta_{x_i}, $$
in which $\delta_x$ denotes
  the Dirac measure at $x$, and for an index $i \in \mathcal{I}$, 
      $x_i = \frac{\int_{I_i}x d\mu(x)}{\mu(I_i)}$. 
Then $\mu'\in\bM_P$ and
\be \int_\bK\lam(x)\mu'(\id x)\leq \int_\bK\lam(x)\mu(\id x).\label{inf_dom}\ee
\end{lemma}
 
\textit{Proof of lemma:} The inequality \eqref{inf_dom} follows from the conditional Jensen inequality. A direct computation shows that $\mu'$ satisfies \eqref{mucond}, so that  $\mu'\in\bM_P$.\hfill$\square$

It remains now to understand when there is existence in the dual problem. We exclude the case when $\lam_T$ is affine on some $[z,\infty)$ which is tedious. We characterise the existence of a dual minimiser in terms of properties of the solution to the primal problem and also present a set of sufficient conditions. Of the conditions given, (i) would never be encountered in practice (it implies the existence of free call options) and (iii),(iv) depend only on the function $\lam_T$ and not on the put prices $P_i$. The examples presented in Section \ref{sec-ex} show that if these conditions fail there may still be existence, but this will now depend on the $P_i$. Condition (ii) is closer to being necessary and sufficient, but is not stated in terms of the basic data of the problem. 
\begin{proposition}\label{prop:LBdualexist} Assume $\lam_T$ is not affine on some half-line $[z,\infty)$. Then, 
in the setup of Proposition \ref{lb}, the existence of a minimiser in the dual problem $D_{\mathrm{LB}}$ fails if and only if 
\begin{equation}\label{eq:dual_equiv}
\nmax=\infty\quad \textrm{and}\quad \left.X\da_T\right|_{S_T=s}=\phi\da\Gamma_T s+\psi\da D_T^{-1}<\lam_T(s), \textrm{ for all }s\geq K_n.
\end{equation}
In particular, each of the following is a sufficient condition for existence of  a minimiser in $D_{\mathrm{LB}}$:
\begin{enumerate}
\item[(i)] $\nmax<\infty$;
\item[(ii)] we have
\begin{equation}\label{eq:dual_iff}
\left.X\da_T\right|_{S_T=K_n}=\phi\da\Gamma_TK_n+\psi\da D_T^{-1}<\lam_T(K_n)\quad\textrm{and}\quad\lim_{s\to\infty}\lam_T(s)-\phi\da\Gamma_Ts=\infty;
\end{equation}
\item[(iii)] for any $y<\lam_T(K_n)$ there is some $x>K_n$ such that the point $(K_n,y)$ lies on a support line to $\lam_T$ at $x$;
\item[(iv)] $\lam_T$ satisfies
\be\int_0^\infty x\lam_T''(\id x)=+\infty\ .\label{eq:lambda_dual_cond}\ee
\end{enumerate}
\end{proposition}
\textit{Proof:} We consider two cases. 
\medskip

\emph{Case 1:} $\nmax\leq n$, i.e. condition (i) holds. In this case the support of any measure $\mu\in\bM_P$ is contained in the finite union $I_{\nmin+1}\cup\cdots\cup I_{\nmax}$ of bounded intervals, and $\sum_{i=\nmin+1}^{\nmax}\mu(I_i)=1$. Further, $p_{\nmax-1}>k_{\nmax}-1$ and \eqref{mucond_b} together imply that $\mu(I_{\nmax})>0$.
    Let $\mu_j, j=1,2,\ldots$ be a sequence of measures such that
\[ \int_{\bK} \func(x) \mu_j(\id x) \to \inf_{\mu \in \bM_P} \left\{
  \int_{\bK}\func(x) \mu(\id x) \right\}\quad\mathrm{as\,\,} j\to\infty.\]
By Lemma \ref{l1}, we may and do assume that each $\mu_j$ is atomic with at most one atom per interval. We denote $w^i_j=\mu_j(I_i)$ and let $x^i_j$ denote the location of the atom in $I_i$. For definiteness, let $x^i_j=\Delta$, where $\Delta$ is some isolated point, if $\mu_j$ has no atom in $I_i$, i.e., if $w^i_j=0$. Let $A$ be the set of indices $i$ such that $x^i_j$ converges to $\Delta$ as $j\to\infty$, i.e. $x^i_j\not=\Delta$ for only finitely many $j$, and let $B$ be the complementary set of indices. Then there exists $j^*$ such that $\sum_{i\in B}w^i_j=1$ for $j>j^*$. Because the $w^i_j$ and $x^i_j$ are contained in compact intervals there exists a subsequence $j_k,\, k=1,2\ldots$ and points $w^i_*, x^i_*$ such that $w^i_{j_k}\to w^i_*$ and $x^i_{j_k}\to x^i_*$ as $k\to\infty$. It is clear that $\sum_{i\in B}w^i_*=1$ and that
\[V^L_D= \lim_{k\to\infty}\int_{\bK}\lam(x)\mu_{j_k}(\id x)=\int_\bK\lam(x)\mu\da(\id x)\]
where $\mu\da(\id x)=\sum_{i\in B}w^i_*\delta_{x^i_*}(\id x)$. Similarly, the integrals of $1,x$ and $[k_i-x]^+$ converge, so $\mu\da\in\bM_P$. Finally, since the intervals $I_i$ are open on the right, it is possible that $x^i_*\in I_{i+1}$. If also $x^{i+1}_*\in I_{i+1}$ we can invoke Lemma \ref{l1} to conclude that this 2-point distribution in $I_{i+1}$ can be replaced by a 1-point distribution without increasing the integral. Thus $\mu\da$ retains the property of being an atomic measure with at most one atom per interval. We have existence in the dual problem and, by the arguments above, $X\da_T=\lam_T(S_T)$ $\id \mu\da$-a.s. 

\emph{Case 2:} $\nmax=\infty$. 
Let $X\da_T(s)$ denote the payoff of the portfolio $X\da$ at time $T$ when $S_T=s$. Suppose that a minimiser $\mu\da$ exists in $D_{\mathrm{LB}}$, which we may take as atomic. Since $\nmax =\infty$ there exists $x\in (k_n,\infty)$ such that $\mu\da(\{x\})>0$ and hence, by \eqref{impl}, $X\da_T(F_Tx)=\lam(F_Tx)$. This shows that \eqref{eq:dual_equiv} fails since $F_Tx>K_n$. 

It remains to show the converse: that if $\nmax=\infty$ but \eqref{eq:dual_equiv} fails then a minimiser $\mu\da$ exists. Given our assumption on $\lam_T$, if \eqref{eq:dual_equiv} fails then $K\da:=\sup\{s\geq 0: X\da_T(s)=\lam_T(s)\}\in[K_n,\infty)$. 
We continue with analysis similar to Case 1 above. Here we have intervals $I_{\nmin}\clc I_n$ covering $\bK_n=[k_{\nmin}, k_n)$, but also a further interval $I_{n+1}=[k_n,\infty)$ which does not necessarily have zero mass. Since $I_{n+1}$ is unbounded, the argument above needs some modification. We take a minimizing sequence of point mass measures $\mu_j$ as in Case 1. If $x_j^{n+1}$ converge (on a subsequence) to a finite $x^{n+1}_*$ then we can restrict our attention to a compact $[0,x^{n+1}_*+1]$ and everything works as in Case 1 above. Suppose to the contrary that $\liminf x^{n+1}_j=\infty$, as $j\to\infty$. We first apply the arguments in Case 1 to the sequence $\t{\mu}_j$ of restrictions of $\mu_j$ to $\bK_n$. Everything is the same as above except that now $\sum_iw^i_j\leq 1$. A subsequence converges to a sub-probability measure $\t{\mu}\da$ on  $\bK_n$, equal to a weighted sum of Dirac measures as above. 
 Define $\mu\da=\t{\mu}\da$ if $\iota_1 \equiv\t{\mu}\da(\bK_n)=1$ and otherwise $\mu\da=\t{\mu}\da+(1-\iota_1)\delta_x$, where $x\in I_{n+1}$ is to be determined. Whatever the value of $x$, $\mu\da$ satisfies conditions \eqref{mucond_a} and \eqref{mucond_c} (the put values depend only on $\t{\mu}\da$). 

Since $\mu_j\in\bM_P$, $\sum_{i=1}^{n+1}w^i_jx^i_j=1$ and in particular
\[ 1-\sum_{i=1}^nw^i_jx^i_j\geq w_j^{n+1}k_n.\]
Taking the limit along the subsequence we conclude that
\be 1-\sum_{i=1}^nw^i_*x^i_*\geq k_n\left(1-\sum_{i=1}^nw^i_*\right)=k_n(1-\iota_1).\label{w*}\ee
By the convergence argument of Case 1, $\iota_2\equiv \int_0^{k_n}x\mu\da(\id x)\leq1$. If $\iota_1<1$ we have only to choose $x=(1-\iota_2)/(1-\iota_1)$ to ensure that the `forward' condition \eqref{mucond_b} is also satisfied. The inequality \eqref{w*} guarantees that $x\geq k_n$. Thus $\mu\da\in\bM_P$ and, since $\nmax=\infty$, $x>k_n$ and hence $K\da>K_n$.

We now show that the complementary case $\iota_1=1$ contradicts $K\da\in [K_n,\infty)$. Indeed, if $\iota_1=1$ we have, since $\nmax=\infty$,
\[ k_n-1< p_n=\int_0^\infty [k_n-x]^+\mu\da(\id x)=\int_0^{k_n}(k_n-x)\mu\da(\id x)=k_n-\sum_{i=1}^nw^i_*x^i_*\ .
\]
 Observe that then
\begin{equation}\label{eq:cnv_lastatom}
w^{n+1}_jx^{n+1}_j =1 - \sum_{i=1}^nw^i_jx^i_j \xrightarrow[j\to\infty]{} 1-\sum_{i=1}^nw^i_*x^i_*=p_n-k_n+1>0.
\end{equation}
Since $K\da\in [K_n,\infty)$, taking $\gamma=\lam_T'(K\da+1)-{X^{\dag}_T}'(K\da+1)=\lam_T'(K\da+1)-\phi\da\Gamma_T>0$, we have
\begin{equation}\label{eq:righttail}
 \lam_T(s) - X\da_T(s)\geq \gamma(s-K\da+1), \quad \forall s\geq K\da+1.
\end{equation}
Define now a new function $\tilde\lam_T$ by
$$\tilde \lam_T(s)=\lam_T(s)\mathbf{1}_{s\leq K\da}+(\phi\da\Gamma_T s + \psi\da D_T^{-1})\mathbf{1}_{s>K\da}$$
so that we have $X\da\leq \tilde\lam_T(S_T)\leq \lam_T(S_T)$ and, by definition, for $S_T> K\da$ we have $X\da= \tilde\lam_T(S_T)< \lam_T(S_T)$. It follows that $X\da$ is also the most expensive subreplicating portfolio for $\tilde\lam_T(s)$ and hence the primal and dual problems for $\lam$ and for $\tilde\lam(x)=\frac{1}{F_T}\tilde\lam_T(F_Tx)$ have all the same value. Writing this explicitly and using \eqref{eq:righttail} and\eqref{eq:cnv_lastatom} gives:
\begin{equation*}
\begin{split}
0=\lim_{j\to\infty}\int_0^\infty (\lam(x)-\tilde\lam(x))\mu_j(\id x)&=\lim_{j\to\infty} (\lam(x^{n+1}_j)-\tilde\lam(x^{n+1}_j))w^{n+1}_j\\
&\geq \lim_{j\to\infty}\gamma(x^{n+1}_j-K\da-1)w^{n+1}_j=\gamma(p_n-k_n+1)>0,
\end{split}
\end{equation*} a contradiction. 

We turn to showing that each of $(i)$--$(iv)$ is sufficient for existence of a dual minimiser $\mu\dag$. \footnote{In fact for this part of the Proposition we do not need to impose any additional conditions on $\lam_T$ apart from convexity.}
Obviously $\nmax<\infty$ is sufficient as observed above. We now show that either of (ii) or (iii) implies that $K\da\in [K_n,\infty)$ and that \eqref{eq:righttail} holds. In the light of the arguments above this will be sufficient. 
$X\da_T(s)$ is linear on $[K_n,\infty)$ and $\lam_T$ is convex hence the difference of the two either converges to a constant or diverges to infinity. In the latter case the difference grows quicker than a linear function, more precisely \eqref{eq:righttail} holds. The former case is explicitly excluded in (ii) and is contradictory with (iii) as the line $\ell(s):=\phi\da\Gamma_T s + \psi\da D_T^{-1}+\lim_{u\to\infty}(\lam_T(u)-X\da_T(u))$ is asymptotically tangential to $\lam_T(s)$ as $s\to\infty$ and hence the condition in (iii) is violated for any $y<\ell(K_n)$. In particular, it follows that $K\da<K_n$ then we can add to $X\da$ a positive number of call options with strike $K_n$ and obtain a subreplicating portfolio with an initial cost strictly larger (recall that $\nmax=\infty$) than $X\da$ which contradicts the optimality of $X\da$. This completes the proof that either (ii) or (iii) is sufficient for existence.

Finally, we argue that (iii) and (iv) are equivalent. In (iii) the point $x$ satisfies
\[ \lam_T(x)+\xi(K_n-x)=y, \quad \textrm{ for some }\xi\in [\lam_T'(x-),\lam_T'(x+)].\]
This equation has a solution $x$ for all $y<\lam_T(K_n)$ if and only if $\lim_{x\to\infty} \lam_T(x)-x\lam_T'(x)=-\infty$. The equivalence with (iv) follows integrating by parts
\[ \int_{K_n}^\infty x\lam_T''(\id x) = \Big(x\lam_T'(x)-\lam_T(x)\Big)\Big|_{K_n}^\infty.\]
\hfill$\square$


\subsubsection{Examples with one put option}\label{sec-ex} 
We consider now examples in which just one put option price is specified. We illustrate different cases when existence in the dual problem holds or fails. For simplicity assume all prices are normalised, i.e. $D_T=F_T=1$, and we are given only a single put option with strike $k=1.2$. The convex function takes the form $\lam(x)=1/x+ax^b$. Computation of the most expensive sub-hedging portfolio can be done by a simple search procedure. 

Consider first the case $a=0$, so that $\lam(x)=1/x$. The results are shown in Figure \ref{fig} with data shown in Table \ref{table}: $p$ is the put price, $x_0,x_1$ the points of tangency, $w_0,w_1$ the implied probability weights on  $x_0,x_1$  and $\psi, \phi$ the units of, respectively, cash and forward in the portfolio.

$p=0.4$ is a `regular' case: we have tangent lines at $x_0,x_1$ and the solution to the dual problem puts weights 8/9, 1/9 respectively on these points. As $p$ increases it is advantageous to include more puts in the portfolio, so $x_1$ increases. At $p=0.6$ we reach a boundary case where $\psi=\phi=0$, and the put is correctly priced by the Dirac measure with weight 1 at $x_0=0.6=k/2$. Obviously, this measure does not correctly price the forward, but it does correctly price the portfolio since $\phi=0$. When $p>0.6$, the only way to increase the put component further is to take $\psi<0$ (and then clearly the optimal value of $\phi$ is 0.) When $p=0.7$ the optimal value is $\psi=-0.8$ and we find that in this and every other such case the implied weight is $w_0=1$, as the general theory predicts.
\begin{table}[ht]
\begin{center}
\begin{tabular}{|c|c|c|c|c|c|c|c|}\hline
    $p$   &    $x_0$     & $x_1$ &     value   &$\psi$ &$\phi$      &$w_0$      &$w_1$\\ \hline
   0.4 &    0.75& 3  &     1.2222 &0.6667&-0.1111& 0.8889 &0.1111\\
   0.6 &    0.6& -  &     1.6667 &0     &0  &1.00 &-\\
   0.7 &    0.5& -  &     2.00 &-0.8&0   & 1.00 &-\\ \hline
\end{tabular}
\vspace{3mm}
\caption{\label{table} Data for Figure \ref{fig}}
\end{center}
\end{table}

\begin{figure}[h]
\begin{center}
\includegraphics[scale=0.4]{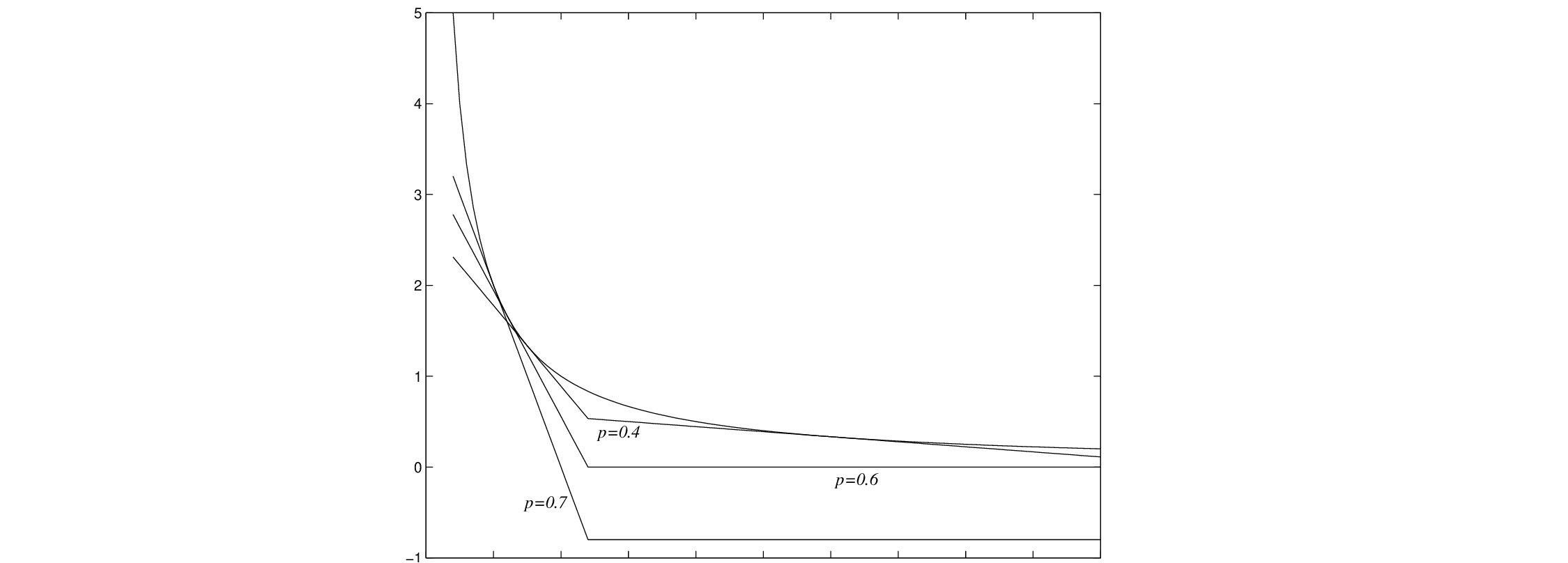}
\caption{\label{fig}Most expensive subhedging portfolios for $\lam(x)=1/x$ given a single put option with strike 1.2 and prices $p=0.4, 0.6, 0.7$}
\end{center}
\end{figure}

Next, take $a=0.25, b=1$, so that $\lam(x)=1/x+0.25x$, has its minimum at $x=2$ and is asymptotically linear. We take $p=0.7$. Taking the expectation with respect to the limiting measure $\mu\da$ gives the value of the lower sub-hedge in Figure \ref{fig-2}, but this is not optimal: we can add a maximum number 0.25 of call options, which have positive value, giving the upper sub-hedge in Figure \ref{fig-2}. This is optimal, but does not correspond to any dual measure.
\begin{figure}[h]
\begin{center}
\includegraphics[scale=0.58]{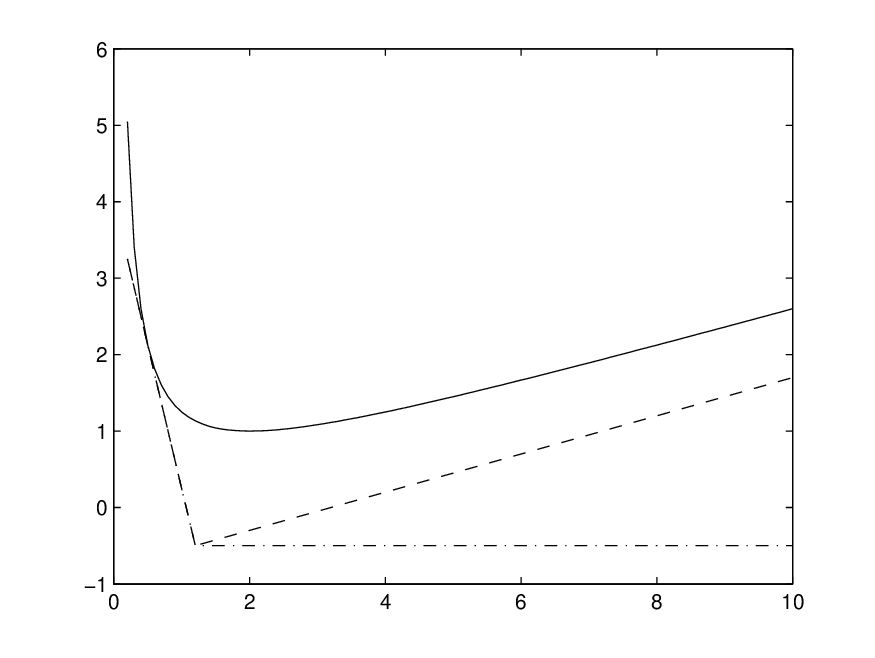}
\vspace{-5mm}\caption{\label{fig-2}The most expensive subhedging portfolio (upper dashed line) for $\lam(x)=1/x+0.25x$ given a single put option with strike 1.2 and price $p=0.7$. The lower dashed line is the portfolio priced at $\int \lam(x)\mu\da(\id x)$ which is suboptimal.}
\end{center}
\end{figure}

Finally, let $a=0.0625,\, b=2$, giving $\lam(x)=1/x+0.0625x^2$. The minimum is still at $x=2$ but $\lam$ is asymptotically quadratic. This function satisfies condition (iii) of Proposition \ref{prop:LBdualexist} and there is dual existence for every arbitrage-free value of $p$. The optimal portfolio for $p=0.7$ is shown in Figure \ref{fig-3}.
\begin{figure}[h]
\begin{center}
\includegraphics[scale=0.58]{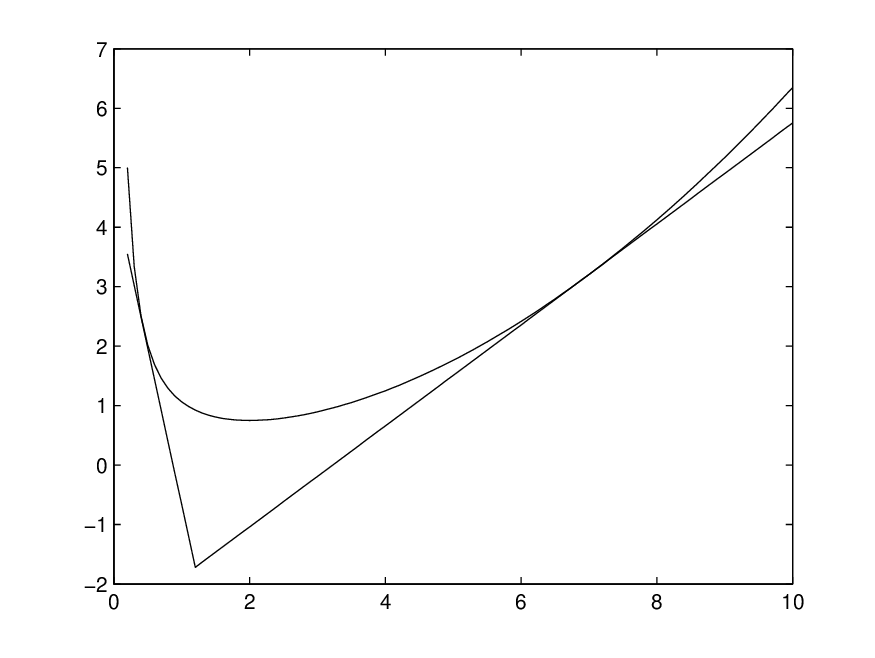}
\vspace{-5mm}\caption{\label{fig-3}Most expensive subhedging portfolios for $\lam(x)=1/x+0.0625x^2$ given a single put option with strike 1.2 and price $p=0.7$.}
\end{center}
\end{figure}

\subsection{Upper bound}\label{sec-upper}
 To compute the cheapest super-replicating portfolio we have to solve the linear program
\[ P_\mathrm{UB}:\, \inf_{\by \in \bR^m}
\by\tr\bb \quad \textrm{subject to} \quad
\by\tr\ba(x) \geq \func(x)\  \forall x \in \bK.\]
The corresponding dual program is 
\[ D_{\mathrm{UB}}:\ \sup_{\mu\in\bM} \int_{\bK} \func(x) \mu( \id x)\quad\mathrm{subject\, to}\quad\int_{\bK}\ba(x)\mu(\id x)=\bb,\] 
where, by Lemma \ref{support}, we may replace $\bK$ by $\bR^+$.
By \eqref{tf} we have for $x\in\bR^+$
\[\lam(x)=\lam(1)+\lam'(1+)(x-1)+\int_{(0,1]}[k-x]^+\lam''(\id k)+\int_{(1,\infty)} [x-k]^+\lam''(\id k).\]
Consider $\mu\in \bM_P$ and let $p_\mu(k)=\int [k-x]^+\mu(\id x)$, $c_\mu(k)=\int [x-k]^+\mu(\id k)$ be the (normalised) prices of puts and calls. Integrating the above against $\mu$ gives
\begin{equation}\label{eq:lam_ub}
\int \lam(x)\mu(\id x)=\lam(1)+\int_{(0,1]}p_\mu(k)\lam''(\id k)+\int_{(1,\infty)} c_\mu(k)\lam''(\id k).
\end{equation}
Recall that $c_\mu(k)=p_\mu(k)+1-k$ and hence maximising in $c_\mu(k)$ or in $p_\mu(k)$ is the same. $p_\mu(k)$ is a convex function dominated by the linear interpolation of points $(k_i,p_i)$, $i=0,1\clc n$, extended to the right of $(k_n,p_n)$ with slope $1$. Since we assume the given put prices do not admit weak arbitrage, it follows from \citep{dh05} (see also Proposition \ref{prop:DH} above) that this upper bound is attained either exactly or in the limit. More precisely, if  $\nmax =\infty$ one can take $\mu_z\in \bM_P$ supported on $\{k_{\nmin},\clc,k_n,z\}$, for $z$ large enough, which attain the upper bound on $[0,k_n]$ and asymptotically induce the upper bound on $(k_n,\infty)$ as $z\to \infty$. It follows from \eqref{eq:lam_ub} that the value of the dual problem is $V^U_D=\lim_{z\to\infty}\int_{\bK}\lam(x)\mu_z(\id x)$. If $\nmax\leq n$ one can take $z=k_{\nmax}$ and $\mu_{k_{\nmax}}$ attains the upper bound. It follows from \eqref{eq:lam_ub} that then $V^U_D=\int_{\bK}\lam(x)\mu_{k_{\nmax}}(\id x)$.

From this observations, one expects that $\check{\by}$ -- the solution to the primal problem -- will correspond to a (normalised) portfolio $\check{\by}\tr\ba(x)$ which linearly interpolates $(k_i,\lam(k_i))$, $i=k_{\nmin}\clc n\land \nmax$ and (if $\nmax=\infty$) extends linearly to the right as to dominate $\lam(k)$. 
The function $x\mapsto\by\tr\ba(x)$ is piecewise linear with a finite number of pieces and no such function can majorize the convex function $\lam$ over $\bR^+$ unless $\lam(0)<\infty$ and $\lam'(\infty)=\gamma<\infty$. In general we impose:
\be \begin{split}(a)&\quad\nmin>0 \textrm{ or } \Big(\nmin=0 \textrm{ and }\lam(0)<\infty\Big)\qquad and\\
(b)&\quad\nmax\leq n \textrm{ or }\Big(\nmax=\infty\textrm{ and }  \lam'(\infty)<\infty\Big)\ .
\end{split} \label{ub_cond} \ee
We have the following result. 
\begin{proposition}\label{prop:UB}
If condition \eqref{ub_cond} holds then there exists a solution $\check{\by}$ to the linear program $P_{\mathrm UB}$. The function $\check{\by}\tr\ba(x)$ is the linear interpolation of the points $(k_{\nmin},\lambda(k_{\nmin})),\clc(k_{n\land \nmax} ,\lam(k_{n\land \nmax}))$ together with, if $\nmax =\infty$, the line $l(x)=\check{\by}\tr\ba(k_n)+(x-k_n)\lam'(\infty)$ for $x\geq k_n$. Primal and dual problem have the same value $V^P_U=V^D_U$ and the existence of a maximiser in the dual problem fails if and only if $\nmax=\infty$ and $\lam$ is not affine on $[k_n,\infty)$.

If the condition \eqref{ub_cond} is not satisfied, there is no feasible solution and $V^D_U=\infty$. 
\end{proposition}
\proof
As argued above, \eqref{ub_cond} is a necessary condition for existence of a feasible solution. Suppose, for example, that $\nmax=\infty=\lam'(\infty)$. Then $(z-k_n)\mu_z(\{z\})=c_n$ and hence $z\mu_z(\{z\})\to c_n$ as $z\to\infty$. Together with $\lam'(\infty)=\infty$ this implies that $V^D_U=\infty$. Other cases are similar.

Suppose \eqref{ub_cond} holds and first consider the case when $\nmax=\infty$. $\lam$ is bounded on $[k_{\nmin},k_n]$ and the linear interpolation is well defined as is the extension beyond $k_n$. Further, there exists some constant $\delta$ such that $\cy\tr\ba(x)-\lam(x)\leq \delta$ for all $x\in \bR^+$. The weight $\cy_{2+j}$ on the $j$th put option is the change in slope at $k_{\nmin + j}$, the `underlying' weight $\by_2$ is equal to $\gamma$, and at $x=k_n$ we have $\lam_n=\cy_1+\cy_2k_n$, so the `cash' weight is $\cy_1=\lam_n-k_n\gamma$. The value of the objective function is, by definition,
\[ \cy\tr\bb=\int_{\bR^+}\cy\tr\ba(x)\mu(\id x),\]
for any $\mu\in \bM_P$. In particular, since $\cy\tr\ba(k_i)=\lam(k_i)$, taking $z$ large enough, we have
\begin{equation}
\begin{split}
\cy\tr\bb^0&=\int_{\bR^+}\cy\tr\ba(x)\mu_z(\id x)=\sum_{i=\nmin}^{\nmax}\lam(k_i)\mu_z(\{k_i\})+\cy\tr\ba(z)\mu_z(\{z\})\\&=\int_{\bR^+}\lam(x)\mu_z(\id x)+(\cy\tr\ba(z)-\lam(z))\mu_z(\{z\})\leq \int_{\bR^+}\lam(x)\mu_z(\id x) + \delta\mu(\{z\}).
\end{split}
\end{equation}
Recall that $\int x \mu_z(\id x)=1$ and in particular $\mu(\{z\})\to 0$ as $z\to \infty$. Taking the limit as $z\to\infty$ in the above, we conclude that $\cy\tr\bb^0\leq V^U_D$. The basic inequality $V^U_D\leq V^U_P$ between the primal and dual values then implies that  $\check{\by}$ is optimal for $P_{\mathrm{UB}}$ and $V^U_D = V^U_P$. The existence of the solution to $D_{\mathrm{UB}}$ fails unless there exists $z\geq k_n$ with $\cy\tr\ba(z)=\lam(z)$, which happens if and only if $\lam(z)$ is affine on $[k_n,\infty)$.

When $\nmax\leq n$ the arguments are analogous, except that now we need to ensure $\by\tr\ba(x) \geq \func(x)$ only for $x \in \bK= [k_{\nmin},k_{\nmax}]$. The dual problem has a maximiser as observed in the remarks above the Proposition. The primal problem also has a solution $\cy$ but it is not unique. Indeed, let $\by^0=(-k_n,1,0\clc 0,1)$ and observe that ${\by^0}\tr\bb=0$ and ${\by^0}\tr\ba(x)\equiv 0$ for $x\in \bK$. In consequence, we can add to $\cy$ multiples of $\by^0$ without affecting its performance for $P_{\mathrm{UB}}$.
\hfill$\square$

We can now summarize the results for the cheapest super-replicating portfolio  as in \eqref{XT},\eqref{X0}. The difference with the above is that we need to ensure super-replication for all possible values of $S_T$ and not only for $S_T\in [K_{\nmin},K_{\nmax}]$. Since the payoff $X_T$, as a function of $S_T$, is piecewise linear with a finite number of pieces it is necessary that $\lam_T$ satisfies
\begin{equation}\label{ub_cond2}
\lam_T(0)<\infty\quad \textrm{and} \quad \lam'_T(\infty)=\gamma<\infty.
\end{equation}
Under this condition, the above results show that the cheapest super-replicating portfolio has a payoff which linearly interpolates $(K_i,\lam_T(K_i))$, $i=0\clc n$ and extends to the right of $K_n$ with slope $\lam_T'(\infty)$.
\begin{proposition}\label{ub2} If \eqref{ub_cond2} holds then there is a cheapest super-replicating portfolio $(\psi^*,\phi^*,\pi_i^*)$ of the European option with payoff $\lam_T(S_T)$ whose initial price is
\[ X^*_0= \sup_zD_T\int_{\bR^+}\lam_T(F_Tx)\mu_z(\id x).\]
The underlying component is $\phi^*=\gamma\Gamma_T$, the cash component is $\psi^*=D_T(\lam_T(K_n)-\gamma K_n)$ and the option components are 
$$\pi_i^*=\frac{\lam_T(K_{j+1})-\lam_T(K_{j})}{K_{j+1}-K_j}-\frac{\lam_T(K_{j})-\lam_T(K_{j-1})}{K_{j}-K_{j-1}}.$$
If condition  \eqref{ub_cond2} is not satisfied, there is no super-replicating portfolio.
\end{proposition}


\subsection{Arbitrage conditions}\label{sec-arb} 
With the above results in hand we can state the arbitrage relationships when a European option whose exercise value at $T$ is a convex function $\lam_T(S_T)$ can be traded at time $0$ at price $P_\lam$ in a market where there already exist traded put options, whose prices $P_i$ are in themselves consistent with absence of arbitrage. Recalling the notation of Propositions \ref{lb} and \ref{ub2}, $X_0\da$ and $X^*_0$ are respectively the setup costs of the most expensive sub-replicating and cheapest super-replicating portfolios, with $X^*_0=+\infty$ when no super-replicating portfolio exists.

  \begin{theorem} \label{convex_payoff} 
  Assume the put prices do not admit a weak arbitrage. Consider a convex function $\lam_T$ and suppose that if $\lam_T$ is affine on some half--line $[z,\infty)$ then it is strictly convex on $[0,z)$. The following are equivalent:
  \begin{enumerate}
\item The prices $P_\lam,P_1\clc P_n$ do not admit a weak arbitrage.
\item There exists a market model for put options in which $P_\lam= D_T\bE[\lam_T(S_T)]$,
\item The following condition \eqref{c1} holds and either $P_\lam\in(X_0\da,X_0^*)$, or $P_\lam=X_0\da$ and existence holds in $D_{\mathrm{LB}}$, or $P_\lam=X_0^*<\infty$ and existence holds in $D_{\mathrm{UB}}$.
\end{enumerate}
\be P_2 > \frac{K_2}{K_1}P_1\ \textrm{if $\nmin=0$ and $\func$ is
  unbounded at the origin.}\label{c1}\ee
If \eqref{c1} holds and $P_\lam\notin[X_0^{\da},X_0^*]$ then there is a model-independent arbitrage. If \eqref{c1} holds and $P_\lam=X\da_0$ or $X_0^*$ and existence fails in $D_{\mathrm{LB}}$ or in $D_{\mathrm{UB}}$ respectively, or if \eqref{c1} fails, then there is a weak arbitrage.
\end{theorem}
\begin{remark}\label{rk:affine_mod}
We note that the robust pricing and hedging problem solved above is essentially invariant if $\lam_T$ is modified by an affine factor. More precisely, if we consider a European option with payoff $\lam_T^1(S_T)=\lam_T(S_T)+\phi S_T + \psi$ and let $P_{\lam^1}$ denote its price then the prices $P_1\clc P_n,P_\lam$ are consistent with absence of arbitrage if and only if $P_1\clc P_n,P_{\lam^1}=P_\lam+\phi S_0/\Gamma_T+\psi D_T$ are.
\end{remark}
\proof Suppose first that condition \eqref{c1} holds. We saw in the proof of Proposition \ref{lb} that this condition (under its equivalent form $p_2 > (k_2/k_1)p_1$) guarantees the existence of a sub-replicating portfolio with value $X_0\da$. If $P_\lam\in(X_0\da,X_0^*)$ then there exists $\epsilon>0$ such that $P_\lam\in(X_0\da+\epsilon, X^*_0-\epsilon)$ and, since there is no duality gap, there are measures $\mu_1,\mu_2\in\bM_P$ such that $D_T\bE_{\mu_1}[\lam_T(S_T)]<X_0\da+\epsilon$ and $D_T\bE_{\mu_2}[\lam_T(S_T)]>X_0^*-\epsilon$.
A convex combination $\mu$ of $\mu_1$ and $\mu_2$ then satisfies $D_T\bE_{\mu}[\lam_T(S_T)]=P_\lam$ and one constructs a market model, for example by using Skorokhod embedding as explained in Section \ref{sec:probform} above. If existence holds in $D_{\mathrm{LB}}$ then it was shown in the proof of Proposition \ref{lb} that the minimizing measure $\mu\da$ satisfies
\[ X\da_0=D_T\int_\bK\lam_T(F_Tx)\mu\da(\id x),\]
so that if $P_\lam=X\da_0$ then $\mu\da$ is a martingale measure that consistently prices the convex payoff $\lam_T$ and the given set of put options. The same argument applies on the upper bound side. 

Next, suppose that condition \eqref{c1} holds and $P_\lam=X_0\da$ but no minimizing measure $\mu$ exists in $D_{\mathrm{LB}}$. 
Let $\cM$ be a model and $\hat{\mu}$ the distribution of $S_T$ under $\cM$.  We can, at zero initial cost, buy $\lam_T(S_T)$ and sell the portfolio $X\da_T$ and this strategy realizes an arbitrage under $\cM$ if $\hat{\mu}(\{\lam_T(S_T)>X_T\da\})>0$. Suppose now that $\hat{\mu}(\{\lam_T(S_T)>X_T\da\})=0$ and consider two cases. First, if $\lam_T$ is not affine on some $[z,\infty)$ then, by Proposition \ref{prop:LBdualexist}, $\nmax=\infty$ and $\lam_T(S_T)>X_T\da$ for $S_T\geq K_n$ so that in particular $\hat{\mu}([K_n,\infty))=0$. A strategy of going short a call option with strike $K_n$ (which has strictly positive price since $\nmax=\infty$) gives an arbitrage since $S_T<K_n$ a.s.\ in $\cM$. Second, suppose that $\lam_T''(x)\equiv 0$ for $x\geq z$ but $\lam_T''(x)>0$ for $x<z$.  
If $\hat{\mu}([K_n,\infty))=0$ then we construct the arbitrage as previously so suppose this does not hold.
Recall the atomic measure $\mu\da$ defined in Case 2 in the proof of Proposition \ref{prop:LBdualexist} and that $\iota_1=1$ as we do not have existence of a minimiser for the dual problem. $\lam(x)$ in \eqref{eq:def_lam} is strictly convex on $[0,\tilde z)$ and linear on $[\tilde z,\infty)$ with $\tilde z=z/F_T$. It is not hard to see that $\mu\da$ has to have an atom in some $x^*_n\in (k_{n-1},k_n)$ and that either $\tilde z=x^*_n$ or else $\tilde z\geq k_n$. Otherwise we could modify $\mu\da$ to obtain a minimiser for the dual problem. Strict convexity of $\lam_T$ on $[0,z)$ implies that there exist at most $n$ points $s_1,\ldots, s_m$ such that $s_i\in (K_{i-1},K_i)$ and $\lam_T(S_T)$ strictly dominates $X\da_T$ for other values of $S_T\leq z$ and hence $\hat\mu([0,z))=\hat\mu(\{s_1,\ldots,s_n\})$. 
It follows that support of $\mu\da$ is a subset of $\{s_1/F_T,\ldots,s_n/F_T=x^*_n\}$. Let $\psi>0$ and consider a portfolio $Y$ with 
$$Y_T=Y_T(S_T)=\sum_{i=1}^n \gamma_i (K_i-S_T)^+ +\psi,\quad \textrm{such that}\quad Y_T(s_j)=0,\ j=1,\dots n.$$
This uniquely specifies $\gamma_i\in \bR$. The payoff of $Y$ is simply a zigzag line with kinks in $K_i$, zero in each $s_i$ and equal to $\psi$ for $S_T\geq K_n$. It follows that $\hat\mu\{Y_T\geq 0\}=1$ and $\hat\mu\{Y_T>0\}>0$ as $\hat{\mu}([K_n,\infty))>0$. However $\mu\da$ prices all put options correctly so that the initial price of $Y$ is
\begin{equation}
\begin{split}
Y_0=&\sum_{i=1}^n \gamma_i P_i + D_T\psi=D_T\sum_{i=1}^n \gamma_i \int_0^{k_i}(K_i-F_Tm)\mu\da(\id m)+D_T\psi\\
=&D_T\int_0^{k_n}\left(\sum_{i=1}^n\gamma_i(K_i-F_Tm)^++\psi\right)\mu\da(\id m)=D_T\int_0^{k_n}Y_T(F_T m)\mu\da(\id m)=0,
\end{split}
\end{equation}
by construction of $Y$ since $\mu\da(\{s_1/F_T,\ldots,s_n/F_T\})=1$ as remarked above, and where we used $\iota_1=1$. It follows that $Y$ is an arbitrage strategy in $\cM$.\\
Now suppose that condition \eqref{c1} holds and $P_\lam=X_0^*<\infty$ and there is no maximising measure in the dual problem $D_{\mathrm{UB}}$.  Then, by Proposition \ref{prop:UB}, $\nmax=\infty$ and $\lam_T(S_T)<X_T^*$ for $S_T>K_n$. Straightforward arguments as in the first case above show that there is a weak arbitrage.\\
 Finally, if $P_\lam<X_0\da$ a model independent arbitrage is given by buying the European option with payoff $\lam_T(S_T)$ and selling the subheding portfolio. This initial cost is negative while the payoff, since $X\da$ subhedges $\lam_T(S_T)$, is non-negative. If $P_\lam>X_0^*$ we go short in the option and long in the superhedge.

Now suppose \eqref{c1} does not hold, so that $P_2/K_2=P_1/K_1$ (this is the only case other than \eqref{c1} consistent with absence of arbitrage among the put options). Consider portfolios with exercise values
\bstar H_1(S_T)&=&[K_2-S_T]^+-\frac{K_2}{K_1}[K_1-S_T]^+\\
H_2(S_T)&=& \lam_T(S_T)-\lam_T(K_2)-\lam'_0(K_2)(S_T-K_2)-\frac{1}{P_1}\left(P_\lam-\lam(K_2)-\lam'_{0}(F_T-K_2)\right)[K_1-S_T]^+,\estar
where $\lam'_0$ denotes the left derivative. The setup cost for each of these is zero, and $H_1(s)>0$ for $s\in(0,K_2)$ while $H_2(s)\to\infty$ as $s\to0$. There is a number $\theta\geq0$ such that $H(s)=\theta H_1(s)+H_2(s)>0$ for $s\in (0,K_2)$. Weak arbitrage is realized by a portfolio whose exercise value depends on a given model $\cM$ and is specified via
\[ X_T(S_T)=\begin{cases}-[K_2-S_T]^+&\mathrm{if\,\, }\bP[S_T\in[0,K_2)]=0\\
H(S_T)&\mathrm{if\,\, }\bP[S_T\in[0,K_2)]>0.\end{cases}\]
This completes the proof.\hfill$\square$
\section{Weighted variance swaps}\label{sec-wvs}
We come now to the second part of the paper where we consider weighted variance swaps. The main idea, as indicated in the Introduction, is to show that a weighted variance swap contract is equivalent to a European option with a convex payoff and hence their prices have to be equal. The equivalence here means that the difference of the two derivatives may be replicated through trading in a model-independent way. In order to formalise this we need to define (continuous) trading in absence of a model, i.e. in absence of a fixed probability space. This poses technical difficulties as  we need to define pathwise stochastic integrals.

One possibility is to define stochastic integrals as limits of discrete sums. The resulting object may depend on the sequence of partitions used to define the limit. This approach was used in \citet{BicWil94} who interpreted the difference resulting from different sequences of partitions as broker's method of implementing continuous time trading order. They were then interested in what happens if they apply the pricing-through-replication arguments on the set of paths with a fixed ($\sigma^2$) realised quadratic variation and wanted to recover Black-Scholes pricing and hedging. However for our purposes the ideas of \cite{BicWil94} are not suitable. We are interested in a much wider set of paths and then the replication of a weighted variance swap combining trading and a position in a European option would depend on the `broker' (i.e.\ sequence of partitions used).
Instead, as in \cite{Lyons:95}, we propose an approach inspired by the work of \citet{Fol81}. We restrict the attention to paths which admit quadratic variation or pathwise local time. For such paths we can develop pathwise stochastic calculus including It\^o and Tanaka formulae. As this subject is self-contained and of independent interest we isolate it in Appendix \ref{sec-pathwise}. Insightful discussions of this topic are found in \cite{BicWil94} and \cite{Lyons:95}.

To the standing assumptions (i)--(iii) of Section \ref{sec:probform} we add another one:\medskip\\
\ni(iv) $(S_t:t\leq T)\in \cL^+$ -- the set of \emph{strictly} positive, continuous functions on $[0,T]$ which admit a finite, non-zero, quadratic variation and a pathwise local time, as formally defined in Definitions \ref{def-qvp},\ref{L-def} and Proposition \ref{LsubQ} of Appendix \ref{sec-pathwise}.\medskip

Thus, our idea for the framework, as opposed to fixing a specific model $\cM$, is to assume we are given a set of possible paths for the price process: $(S_t: t\leq T)\in \pathspc$. This could be, for example, the space of continuous non-negative functions, the space of functions with finite non-zero quadratic variation or the space of continuous functions with a constant fixed realised volatility. The choice of $\pathspc$ is supposed to reflect our beliefs about characteristics of price dynamics as well as modelling assumptions we are willing to take. Our choice above, $\pathspc=\cL^+$, is primarily dictated by the necessity to develop a pathwise stochastic calculus. It would be interesting to understand if an appropriate notion of no-arbitrage implies (iv). A recent paper of \citet{Vov11}, based on a game-theoretic approach to probability, suggests one may exclude paths with infinite quadratic variation through a no-arbitrage-like restriction, an interesting avenue for further investigation.

We introduce now a continuous time analogue of the weighted realised variance \eqref{rv}. Namely, we consider a market in which, in addition to finite family of put options as above, a $w$-weighted variance swap is traded. It is specified by its payoff at maturity $T$:
\begin{equation}\label{eq:wRVpayoff}
RV^w_T - \prv:=\int_0^T w(S_t/F_t) \id \langle \log S\rangle_t - P^{\text{\tiny RV(w)}}_T,
\end{equation}
where $\prv$ is the swap rate, and has null entry cost at time $0$. The above simplifies \eqref{rv} in two ways. First, similarly to the classical works on variance swaps going back to \cite{neuberger_94}, we consider a continuously and not discretely sampled variance swap which is easier to analyse with tools of stochastic calculus. Secondly, the weighting in \eqref{rv} is a function of the asset price $h(S_{t_i})$ and in \eqref{eq:wRVpayoff} it is a function of the ratio of the actual and the forward prices $w(S_t/F_t)$. This departure from the market contract definition is unfortunate but apparently necessary to apply our techniques. In practice, if $\hat{w}(S_t)$ is the function appearing in the contract definition we would apply our results with $w(x)=\hat{w}(S_0x)$, so that $w(S_t/F_t)=\hat{w}((S_0/F_t)S_t)$. Since maturity times are short and, at present, interest rates are low, we have $S_0/F_t\approx 1$. See below and Section \ref{comp} for further remarks.

Our assumption (iv) and Proposition \ref{YfX} imply that $(\log S_t, t\leq T)\in\cL$. Theorem \ref{extended-ito} implies that \eqref{eq:wRVpayoff} is well defined as long as $w\in \Ltwoloc$, we can integrate with respect to $S_t$ or $M_t$ and obtain an It\^o formula. This leads to the following representation.
\begin{lemma} \label{lem:hedge_thm}
  Let $w:\R_+\to [0,\infty)$ be a locally square integrable function and consider a convex $C^1$ function $\lam_w$ with $\lam_w''(a)=\frac{w(a)}{a^2}$. The extended It\^o formula \eqref{ito} then holds and reads
\begin{equation}\label{eq:wRVhedge}
\lam_w(M_T)   = \lam_w(1) +\int_0^{T} \lam_w'(M_u) \id M_u + \frac{1}{2}\int_{[0,T]} w(M_u) \id \langle \ln M \rangle_u.
\end{equation}
\end{lemma}

The function $\lam_w$ is specified up to an addition of an affine component which does not affect pricing or hedging problems for a European option with payoff $\lam_w$, see Remark \ref{rk:affine_mod} above. In what follows we assume that $w$ and $\lam_w$ are fixed.
Three motivating choices of $w$, as discussed in the Introduction, and the corresponding functions $\func_w$, are:
\begin{enumerate} \label{funcs}
\item Realised variance swap: $w \equiv 1$ and $\func_w(x) = -\ln(x)$. In this case there is of course no distinction between $w$ and the contract function $\hat{w}$. 
\item Corridor variance swap: $w(x)  = \mathbf{1}_{(0,a)}(x)$ or $w(x) =
  \mathbf{1}_{(a,\infty)}(x)$, where $0<a<\infty$ and 
$$ \func_w(x) =  \left( -\ln\left(\frac{x}{a}\right) + \frac{x}{a} - 1 \right) w(x).$$ Here we would take $a=b/S_0$ if the contract corridor is $(0,b)$ or $(b,\infty)$
\item Gamma swap: $w(x)=S_0x$ and $ \func_w(x) = S_0(x \ln (x) - x)$.
\end{enumerate}
Clearly, \eqref{eq:wRVhedge} suggests that we should consider portfolios which trade dynamically and this will allow us to link $w$-weighted realised variance $RV^w_T$ with a European option with a convex payoff $\lam_w$. Note however that it is sufficient to allow only for relatively simple dynamic trading where the holdings in the asset only depend on asset's current price. More precisely, we extend the definition of portfolio $X$ from static portfolios as in \eqref{XT}-\eqref{X0} to a class of dynamic portfolios. We still have a static position in traded options. These are options with given market prices at time zero and include $n$ put options but could also include another European option, a weighted variance swap or other options. At time $t$ we also hold $\Gamma_t\phi(M_t)$ assets $S_t$ and $\psi_t/D_t$ in cash. The portfolio is self-financing on $(0,T]$ so that 
\begin{equation}\label{eq:psi_def}
\psi_t:=\phi(M_0)S_0+\psi(0,S_0)+\int_0^t \phi(M_u)\id M_u - \Gamma_t\phi(M_t)S_tD_t,\quad t\in (0,T],
\end{equation}
and where $\phi$ is implicitly assumed continuous and with a locally square integrable weak derivative so that the integral above is well defined, cf.~Theorem \ref{extended-ito}. We further assume that there exist: a linear combination of options traded at time zero with total payoff $Z=Z(S_t:t\leq T)$, a convex function $G$ and constants $\tilde \phi, \tilde \psi$ such that
\begin{equation}\label{eq:admis_cond}
\Gamma_t\phi(M_t)S_t + \psi_t/D_t\geq Z - G(M_t)/D_t+ \tilde \phi\Gamma_t S_t + \tilde \psi/D_t,\quad \forall t\leq T.
\end{equation}
Such a portfolio $X$ is called \emph{admissible}. Observe that, in absence of a model, the usual \emph{integrability} of $Z$ is replaced by \emph{having finite price at time zero}. 
In the classical setting, the admissibility of a trading strategy may depend on the model. Here 
admissibility of a strategy $X$ may depend on which options are assumed to trade in the market. The presence of the term $G(M_t)$ on the RHS will become clear from the proof of Theorem \ref{thm:wVS_arb} below. It allows us to enlarge the space of admissible portfolios for which Lemma \ref{lem:mm->na} below holds.

The two notions of arbitrage introduced in Section \ref{sec:probform} are consequently extended by allowing not only static portfolios but possibly dynamic admissible portfolios as above. All the previous results remain valid with the extended notions of arbitrage. Indeed, if given prices admit no \emph{dynamic} weak arbitrage then in particular they admit no \emph{static} weak arbitrage. And for the reverse, we have the following general result.
\begin{lemma}\label{lem:mm->na}
Suppose that we are given prices for a finite family of co-maturing options\footnote{These could include European as well as exotic options.}. If a market model $\cM$ exists for these options then any admissible strategy $X$ satisfies $\bE[D_TX_T]\leq X_0$. In particular, the prices do not admit a weak arbitrage.
\end{lemma}
\begin{proof}
Let $\cM$ be a market model and $X$ be an admissible strategy. We have $X_T=Z^1+Y_T$, where $Z^1$ is a linear combination of payoffs of traded options and 
$Y_t=\Gamma_t\phi(M_t)S_t + \psi_t/D_t$ satisfies \eqref{eq:admis_cond}. Using \eqref{eq:psi_def} it follows that 
$$D_tY_t=\phi(1)S_0+\psi(0,S_0)+\int_0^t \phi(M_u)\id M_u\geq D_tZ-G(M_t)+ \tilde\phi^+S_0M_t+\tilde \psi.$$
We may assume that $G\geq 0$, it suffices to replace $G$ by $G^+$.
$\cM$ is a market model and in particular $Z$ is an integrable random variable.
Since the traditional stochastic integral and our pathwise stochastic integral coincide a.s.\ in $\cM$, see Theorem \ref{prob_conn}, we conclude that $D_tY_t$ is a local martingale and so is $N_t:=D_tY_t-\tilde\phi^+S_0M_t$. We will argue that this implies $\bE N_t\leq N_0$. Let $\rho_n$ be the localising sequence for $N$ so that $\bE N_{t\land \rho_n} = N_0$. In what follows all the limits are taken as $n\to \infty$. Fatou's lemma shows that $\bE N^+_t\leq \liminf \bE N_{t\land \rho_n}^+$ and $\bE N^-_t\leq \liminf \bE N_{t\land \rho_n}^-$. By Jensen's inequality the process $G(M_t)$ is a submartingale, in particular the expectation is increasing and $\bE G(M_t)\leq \bE G(M_T)$ which is finite since $\cM$ is a market model. Using the Fatou lemma we have
\[ \lim \bE G(M_{t\wedge\rho_n})\leq\bE G(M_t)=\bE \liminf G(M_{t\wedge\rho_n})\leq \liminf \bE G(M_{t\wedge\rho_n}),\]
showing that $\lim \bE G(M_{t\land \rho_n})=\bE G(M_t)$. Observe that $N_t$ is bounded below by $\tilde Z - G(M_t)$, where $\tilde Z$ is an integrable random variable and $G$ is convex. Using Fatou lemma again we can write
$$\bE G(M_t)^++\bE \tilde Z^- - \liminf \bE N^-_{t\land \rho_n}= \liminf \bE\left[G(M_{t\land \rho_n})^++\tilde Z^--N^-_{t\land \rho_n}\right]\geq \bE\left[ G(M_t)^++ \tilde Z^- - N^-_t\right]
$$
which combined with the above gives $\bE N^-_t = \liminf \bE N_{t\land \rho_n}^-$ and in consequence $\bE N_t\leq \liminf \bE N_{t\land \rho_n} = N_0$, as required. This shows that $\bE [D_TY_T]\leq Y_0$. Since in a market model expectations of discounted payoffs of the traded options coincide with their initial prices it follows that $\bE[D_TX_T]\leq X_0\leq 0$.  In particular if $X_T\geq 0$ and $X_0\leq 0$ then $X_T=0$ a.s.\ and the prices do not admit a weak arbitrage. 

\hfill$\square$
\end{proof}

Having extended the notions of (admissible) trading strategy and arbitrage, we can now state the main theorem concerning robust pricing of weighted variance swaps.
 It is essentially a consequence of the hedging relation \eqref{eq:wRVhedge} and the results of Section \ref{sec-cp}. 
\begin{theorem}\label{thm:wVS_arb}
Suppose in the market which satisfies assumptions (i)--(iv) the following are traded at time zero: $n$ put options with prices $P_i$, a $w$--weighted variance swap with payoff \eqref{eq:wRVpayoff} and a European option with payoff $\lam_T(S_T)=F_T\lam_w(M_T)$ and price $P_{\lam_w}$. Assuming the put prices do not admit a weak arbitrage, the following are equivalent
\begin{enumerate}
\item 
The `option' prices (European options and weighted variance swap) do not admit a weak arbitrage.
\item $P_1,\ldots,P_n,P_{\lam_w}$ do not admit a weak arbitrage and   
\begin{equation}\label{eq:vs_euro}
\prv=\frac{2 P_{\lam_w}}{D_TF_T}-2\lam_w(1).
\end{equation}
\item A market model for all $n+2$ options exists.
\end{enumerate}
\end{theorem}
\begin{remark}
It is true that under \eqref{eq:vs_euro} a market model for $P_1,\ldots,P_n,\prv$ exists if and only if a market model for $P_1,\ldots,P_n,P_{\lam_w}$ exists. By Theorem \ref{convex_payoff}, this is yet equivalent to $P_1,\ldots,P_n,P_{\lam_w}$ being consistent with absence of arbitrage. However it is not clear if this is equivalent to $P_1,\ldots,P_n,\prv$ being consistent with absence of arbitrage. This is because the portfolio of the variance swap and dynamic trading necessary to synthesise $-\lam_T(S_T)$ payoff may not be admissible when $\lam_T(S_T)$ is not a traded option. 
\end{remark}
\begin{remark}
The formulation of Theorem \ref{thm:wVS_arb} involves no-arbitrage prices but these are enforced via robust hedging strategies detailed in the proof. They involve the European option with payoff $\lambda_T(S_T)$ which in practice may not be traded and should be super-/sub- replicated using Propositions \ref{lb} and \ref{prop:UB}.
\end{remark}
\begin{remark}\label{rk:Hobson_Klimmek}.
 It would be interesting to combine our study with the results of \cite{HobsonKlimmek:11} [HK] already alluded to in the Introduction. We note however that this may not be straightforward since the European option constituting the static part of the hedge in HK need not be convex and the variance kernels we consider are not necessarily monotone (in the terminology of HK), for example the Gamma swap $y(\log(y/x))^2$. Finally we note that the bounds in HK are attained by models where quadratic variation is generated entirely by a single large jump which is a radical departure from the assumption of continuous paths. Whether it is possible to obtain sharper bounds which only work for ``reasonable" discontinuous paths is an interesting problem. We leave these challenges to future research.
\end{remark}

\begin{proof}
We first show that $2\Longrightarrow 3$. Suppose that $P_1,\ldots,P_n,P_{\lam_w}$ do not admit a weak arbitrage and let $\cM$ be a market model which prices correctly the $n$ puts and the additional European option with payoff $\lam_T(S_T)$. Note that, from the proof of Theorem \ref{convex_payoff}, $\cM$ exists and may be taken to satisfy (i)--(iv). Consider the It\^o formula \eqref{eq:wRVhedge} evaluated at $\tau_n\land T=\inf\{t\geq 0: M_t\notin (1/n,n)\}\land T$ instead of $T$. The continuous function $\lam_w'$ is bounded on $(1/n,n)$, the stochastic integral is a true martingale and taking expectations we obtain
$$\bE[\lam_w(M_{\tau_n\land T})]=\lam_w(1)+\frac{1}{2}\bE\left[\int_0^{\tau_n\land T}w(M_u)\id \langle M\rangle_u\right].$$
Subject to adding an affine function to $\lam_w$ we may assume that $\lam_w\geq 0$. Jensen's inequality shows that $\lam_w(M_{\tau_n\land T}),\,\,n\geq 2$, is a submartingale. Together with the Fatou lemma this shows that the LHS converges to $\bE[\lam_w(M_T)]$ as $n\to \infty$. Applying the monotone convergence theorem to the RHS we obtain
$$\bE[\lam_w(M_T)]=\lam_w(1)+ \frac{1}{2}\bE\left[RV^w_T\right],$$
where either both quantities are finite or infinite. Since $\cM$ is a market model for puts and $\lam_T(S_T)$, we have $\bE[\lam_w(M_T)]=P_{\lam_w}/(D_TF_T)$. Combining \eqref{eq:vs_euro} with the above it follows that $\bE\left[RV^w_T-\prv\right]=0$ and hence $\cM$ is a market model for puts, $\lam_T(S_T)$ and the $w$--wighted variance swap.

Lemma \eqref{lem:mm->na} implies that $3\Longrightarrow 1$. We note also that, if we have a market model $\cM$ for all the puts, $\lam_w(S_T)$ option and the $w$--wighted variance swap then by the above \eqref{eq:vs_euro} holds. Then Lemma \eqref{lem:mm->na} also implies that $3\Longrightarrow 2$.

It remains to argue that $1\Longrightarrow 2$ i.e.\ that if $P_1,\ldots,P_n,P_{\lam_w}$ are consistent with absence of arbitrage but \eqref{eq:vs_euro} fails then there is a weak arbitrage. Consider a portfolio $X$ with no put options, $\psi(0,S_0)=2D_T(\lam_w(1)-S_0\lam'_w(1))+D_T\prv$ and $\phi(m)=2D_T\lam_w'(m)$. The setup cost of $X$ is $X_0=D_T(\prv+2\lam_w(1))$. By assumption, $\lam_w$ is $C^1$ with $\lam_w''(x)=w(x)/x^2\in \Ltwoloc$ so that can apply Theorem \ref{extended-ito}. Then, using \eqref{eq:psi_def} and \eqref{eq:wRVhedge} we obtain
\begin{equation*}
\begin{split}
X_t&=\Gamma_t\phi(M_t)S_t+\psi_t/D_t=\phi(S_0)S_0/D_t+\psi(0,S_0)/D_t+\frac{1}{D_t}\int_0^t\phi(M_u)\id M_u\\
&=\frac{D_T}{D_t}\left(2\lam_w(1)+2\int_0^t \lam_w'(M_u)\id M_u+\prv\right)\\
&=\frac{2D_T}{D_t}\lam_w(M_t)-\frac{D_T}{D_t}\int_{[0,t]}w(S_u)\id\langle \ln M\rangle_u +\frac{D_T}{D_t}\prv\ .
\end{split}
\end{equation*} 
Observing that $1\geq D_T/D_t\geq D_T$ and $\int_{[0,t]}w(S_u)\id\langle \ln M\rangle_u$ is increasing in $t$ it follows that both $X$ and $-X$ are admissible. 
Suppose first that 
\begin{equation}\label{eq:1case_LamVS}
\prv<\frac{2 P_{\lam_w}}{D_TF_T}-2\lam_w(1).
\end{equation}
Consider the following portfolio $Y$: short $2/F_T$ options with payoff $\lam_T(S_T)$, long portfolio $X$ and long a $w$--weighted variance swap. $Y$ is admissible, the initial cost is
$$Y_0=-2P_{\lam_w}/F_T+X_0= -2P_{\lam_w}/F_T+D_T(\prv+2\lam_w(1))<0,$$ while $Y_T=0$ and hence we have a model independent arbitrage. If a reverse inequality holds in \eqref{eq:1case_LamVS} then the arbitrage is attained by $-Y$.

\hfill $\square$
\end{proof}

\section{Computation and comparison with market data}\label{comp}
\subsection{Solving the lower bound dual problem}
When one of the sufficient conditions given in Proposition \ref{prop:LBdualexist} is satisfied, for existence in the lower bound dual problem, then this problem can be solved by a dynamic programming algorithm. We briefly outline this here, referring the reader to \citet{Rav10} for complete details. For simplicity, we restrict attention to the practically relevant case $\nmin=0, \,\nmax=\infty$. Of course, once the dual problem is solved, the maximal sub-hedging portfolio is immediately determined.

The measure $\mu^{\dag}$ to be determined
satisfies
$$ \int \func(x) \mu^{\dag}(\id x) = \inf_{\mu \in \bM_P} \left\{
  \int_{\bR^+}\func(x) \mu(\id x) \right\}.$$
and we recall from Lemma \ref{l1} that we can restrict our search to measures of the form $\mu(\id x)=\sum_{i=1}^{n+1}w_i\delta_{\chi_i}(\id x)$ where $\chi_i\in[k_{i-1},k_i)$, $w_i\geq0, \sum_iw_i=1$. We denote $\zeta_0=0$ and, for $i\geq1$, $\zeta_i=\sum_1^iw_j$, the cumulative weight on the interval $[0,k_i)$.
For consistency of the put prices $r_1, \ldots, r_n$ with absence of arbitrage, Proposition \ref{prop:DH} dictates that
\[ 
  \zeta_i \in A_i=\left[ \frac{r_i - r_{i-1}}{k_i - k_{i-1}},
    \frac{r_{i+1} - r_{i}}{k_{i+1} - k_{i}} \right]\ \textrm{ for
    $1 \leq i < n$}\]
and
\[ \zeta_{n}\in A_{n}=  \left[ \frac{r_{ n} - r_{{ n}-1}}{k_{ n} - k_{{ n}-1}},
    1 \right].\]
Given $\zeta=(\zeta_1\clc\zeta_n)$ (the final weight is of course $w_{n+1}=1-\zeta_n$), the positions $\chi_i$ are determined by pricing the put options.
We find that when $\zeta_{i-1}< \zeta_i$
\begin{eqnarray} \label{x_val}
& &\chi_i = \chi_i(\zeta_{i-1},\zeta_i) = k_i + \frac{\zeta_{i-1}(k_i - k_{i-1}) - (r_i - r_{i-1})}{\zeta_i -
  \zeta_{i-1}}\ \textrm{ for }\ i = 1,
  \ldots,n, \nonumber\\
& & \chi_{n+1} = \chi_{ n+1}(\zeta_{n}) = k_{ n} +\frac{1 + r_{ n} -
    k_{ n}}{1 - \zeta_{ n}} \nonumber
\end{eqnarray}
The measure corresponding to policy $\zeta$ is thus \be \sum_{i=1}^n(\zeta_i -\zeta_{i-1})\delta_{\chi_i(\zeta_{i-1}, \zeta_i)} + (1 -
\zeta_n)\delta_{\chi_{n+1}(\zeta_n)}.  \label{meas_policy} \ee 
It follows that the minimisation problem $\inf_{\mu \in \bM_P}
\int_{\bK} \func(x) \mu(\id x)$ has the same value as
\be v_0= \inf_{\zeta_{1} \in A_{1}}\ldots\inf_{\zeta_{n} \in A_{n}}
  \Bigg\{ \sum_{i=1}^{n}(\zeta_i - \zeta_{i-1}) \func(\chi_i(\zeta_{i-1}, \zeta_i))  + (1 -\zeta_{{n}})\func(\chi_{{n}+1}(\zeta_{n})) \Bigg\} .
 \label{minimizer} \ee
We can solve this by backwards recursion as follows. Define
\be\label{dp} \begin{array}{rcl}V_n(\zeta_n)&=&(1-\zeta_n)\chi_{n+1}(\zeta_n)\\
V_j(\zeta_j)&=&\inf_{\zeta_{j+1}\in A_{j+1},\, \zeta_{j+1}\geq\zeta_j}
\{ (\zeta_{j+1} - \zeta_{j}) \func(\chi_i(\zeta_{j}, \zeta_{j+1}))  + V_{j+1}(\zeta_{j+1})\},\,j=n-1\clc0.\end{array}\ee
Then $V_0(0)=v_0$. For a practical implementation one has only to discretize the sets $A_j$, and then \eqref{dp} reduces to a discrete-time, discrete-state dynamic program in which the minimization at each step is just a search over a finite number of points.  

\subsection{Market data} The vanilla variance swap is actively traded in the over-the-counter
(OTC) markets.  We have collected variance swap and European option
data on the S\&P 500 index from the recent past.  Using put option
prices, the lower arbitrage-bound for the variance swap rate is
computed in each case, and summarised in Table \ref{chwrvtable:data}.
One sees that the traded price of the variance swap frequently lies
very close to the lower bound.  Nonetheless, under our standing assumptions of frictionless markets, all but one of the prices were consistent with absence of
arbitrage. The crash in October 2008 of the S\&P 500, and indeed
the financial markets in general, gave rise to significant increase in
expected variance, which can be seen in the cross-section of data studied. One data point, the 3-month contract on 20/12/2008, lies below our lower bound and, at first sight, appears to represent an arbitrage opportunity. However, this data point should probably be discarded. First, practitioners tell us that this was a day of extreme disruption in the market, and indeed the final column of Table \ref{chwrvtable:data}, giving the number of traded put prices used in the calculations, shows that December 19 and 20 were far from typical days. Second, the fact that the quotes for 19 and 20 December are exactly the same makes it almost certain that the figure for 20 December is a stale quote, not a genuine trade. It is a positive point of our method that we are able to pick up such periods of market dislocation, just from the raw price data.

A further point relates to our discussion in Section \ref{sec-wvs} about the weight function $w$ and its relation to the contract weight $\hat{w}$. This is a moot point here, since $w=\hat{w}=1$, but note from the final column in Table \ref{chwrvtable:data} that Libor rates are generally around $3\%$ (albeit with some outliers), while the S\&P500 dividend yield increased from around $2.0\%$ to around $3.1\%$ over the course of 2008. Since the rate closely matches the dividend yield we have $F_t=S_0$ to a close approximation for $t$ up to a few months. 

\begin{table}[ht]
\caption{Historical variance swap (VS) quotes for the S\&P 500 index and
    the lower bound (LB) for it, implied by the bid prices of liquid European put
    options with the same maturity.  The units for the variance price
    and LB are volatility percentage points, $100 \times \sqrt{P^{\text{\tiny VS}}_T}$, and M stands for months.  The European option price data is courtesy
    of UBS Investment Bank, and the variance-swap data was provided by Peter Carr.}\label{chwrvtable:data}
\begin{center}    
\begin{tabular}{|c|c|c|c|c|c|}
 \hline
Term & Quote date & VS quote & LB & \# puts&Libor\\
\hline
2M & 21/04/2008 & 21.24 & 20.10 & 50&2.79\\
2M & 21/07/2008 & 22.98 & 22.51 & 50&2.79\\
2M & 20/10/2008 & 48.78 & 46.58 & 93&4.06\\
2M & 20/01/2009 & 52.88 & 47.68 & 82&1.21\\
\hline                  
3M & 31/03/2008 & 25.87 & 23.59 & 42&2.78\\
3M & 20/06/2008 & 22.99 & 21.21 & 46&2.76\\
3M & 19/09/2008 & 26.78 & 25.68 & 67&3.12\\
3M & 19/12/2008 & 45.93 & 45.38 & 112&1.82\\
3M & 20/12/2008 & 45.93 & 65.81 & 137&1.82\\
\hline                          
6M & 24/03/2008 & 25.81 & 25.34 & 33&2.68\\
6M & 20/06/2008 & 23.38 & 23.20 & 38&3.10\\
\hline
\end{tabular}
\end{center}
\end{table}
\appendix
\section{The Karlin \& Isii Theorem}\label{sec-app}
Let $a_1, \ldots, a_m, f$ be real-valued, continuous functions on $\bK \subset \mathbb{R}^d$ and let $\mathbf{M}$ denote the collection of all finite Borel measures $\mu$ on $\bK$ fulfilling the integrability conditions $ \int_{\bK} |a_i(\bx)| \mu(\id \bx) < \infty$ for
  $i=1, \ldots, m$.  
  For a fixed vector $\mathbf{b}=(b_1, \ldots, b_m)\tr$, and letting $\mathbf{a}(\bx) = (a_1(\bx), \ldots, a_m(\bx))\tr $ for $\bx \in \bK$, consider the optimisation problem
$$ (P):\, \sup_{\mathbf{y} \in \mathbb{R}^{m}}
\mathbf{y}\tr \mathbf{b}\quad \textrm{s.t.} \quad \mathbf{y}\tr \mathbf{a}(\bx) \leq f(\bx)\, \textrm{ }  \forall \bx \in \bK.$$
Now, define 
$$ (D):\, \inf_{\mu \in \mathbf{M}}  \int_{\bK} f(\bx) \mu( \id
\bx)\quad \textrm{s.t.}\quad \int_{\bK} \mathbf{a}(\bx) \mu( \id \bx) = \mathbf{b},$$
where the constraint should be interpreted as $ \int_{\bK} a_i(\bx) \mu( \id \bx) = b_i,\, i=1, \ldots, m$.
The values of the problems $(P)$ and $(D)$  will respectively
be denoted by $V(P)$ and $V(D)$. Finally, $M_m \subset \R^m$ will denote the moment
cone defined by
\be M_m = \left\{ \tilde\bb = (\tilde b_1, \ldots, \tilde b_m)\tr|\, \tilde b_i = \int_{\bK}a_i(\bx) \mu(\id \bx),\, i=1, \ldots, m,\, \mu \in \mathbf{M} \right\} .\label{mc}\ee

\begin{theorem}\label{chconvthm:isii}
Suppose
\begin{enumerate}
\item $a_1, \ldots, a_m$ are linearly independent over $\bK$,
\item $\mathbf{b} $ is an interior point of $ M_m$ and
\item $V(D)$ is finite.
\end{enumerate}
Then
$V(P) = V(D)$ and $(P)$ has a solution.
\end{theorem}

\textsc{Remarks}
(i) The beauty in the Karlin \& Isii theorem is that the proof draws upon no more than a finite dimensional separating hyperplane theorem.  Proofs can be found in  \citet{glashoff:SIP} and Karlin \& Studden \citet[Chapter XII, Section 2]{karlin_studden}.  We follow the latter below.

(ii) Condition 1 of the theorem ensures that the moment cone $M_m$ of \eqref{mc} is $m$-dimensional, i.e. not contained in some lower-dimensional hyperplane. 

\medskip\emph{Proof of Theorem \ref{chconvthm:isii}.} 
Define the enlarged moment cone
$$ M = \left\{ \tilde \bb =(\tilde b_1, \ldots, \tilde b_{m+1})\tr :\ \tilde \bb=\int_{\bK} \left( \begin{array}{c} \mathbf{a}(\bx) \\ f(\bx) \end{array} \right)\mu(\id \bx),\ \mu \in \mathbf{M} \right\},$$
and let $\bar{M}$ denote its closure. Then $\bar{M}$ is a closed convex cone and, moreover, the vector $(b_1, \ldots, b_m, V(D))$ lies on its boundary.  There exists a supporting hyperplane to $\bar{M}$ through this vector, specified by real constants $\{z_i\}_1^{m+1}$ not all zero, such that
\begin{eqnarray}
\sum_{i=1}^m z_i b_i + z_{m+1} V(D) &=& 0 \label{chappeq:sh}\\
\textrm{and} \quad \sum_{i=1}^{m+1} z_i \tilde b_i & \geq& 0, \quad \forall \tilde \bb \in \bar{M}. \label{chappeq:hs}
\end{eqnarray}
In particular, on considering Dirac measures, one has
\be \sum_{i=1}^m z_i a_i(\bx) + z_{m+1} f(\bx) \geq 0, \quad \forall \bx \in \bK. \label{chappeq:diracs} \ee
We now show that $z_{m+1} > 0$.  Indeed, for any $\delta > 0$, the vector $(b_1, \ldots, b_m, V(D)-~ \delta)$ lies in the half-space complimentary to \eqref{chappeq:hs}.  Therefore
$$ \sum_{i=1}^m z_i b_i + z_{m+1} (V(D) - \delta) < 0, \quad \forall \delta > 0.   $$
This clearly implies $z_{m+1} \geq 0$.  However, $z_{m+1}=0$ is not possible, since this would contradict the assumption that $\mathbf{b}$ lies in the interior $M_m$, which is $m$ dimensional.  Thus, it must be that $z_{m+1}>0$.  Then from \eqref{chappeq:diracs}, it follows that
$$ f(\bx) \geq \sum_{i=1}^m \hat{y}_i a_i(\bx), $$
where $\hat{y}_i := -z_i/z_{m+1}$ for $i = 1, \ldots, m$.  Then \eqref{chappeq:sh} becomes
$$ V(D) = \sum_{i=1}^m \hat{y}_i b_i,$$
and this completes the proof.
\hfill$\square$
\section{Pathwise stochastic calculus}\label{sec-pathwise}\nocite{sondermann}
This section describes a non-probabilistic approach to stochastic calculus, due to \citet{Fol81}, that will enable us to define the continuous-time limit of the finite sums \eqref{rv} defining the realized variance, without having to assume that the realized price function $t\mapsto S_t$ is a sample function of a semimartingale defined on some probability space.
 
 For $T>0$ let $\cT=[0,T]$ with Borel sets $\cB_\cT$. A \emph{partition} is a finite, ordered sequence of times $\pi=\{0= t_0<t_1<\cdots<t_k= T\}$ with \emph{mesh size} $m(\pi)=\max_{1\leq j\leq k}(t_j-t_{j-1})$. We fix a nested sequence of partitions $\{\pi_n, n=1,2,\ldots\}$ such that $\lim_{n\to\infty} m(\pi_n)=0$. All statements below relate to this specific sequence. An obvious choice would be the set of dyadic partitions $t^n_j=jT/2^n, j=0\clc 2^n, n=1,2,\ldots$. 
\begin{definition}\label{def-qvp}
A continuous function $X:~\cT\to\bR$ \emph{has the quadratic variation property} if  the sequence of measures on $(\cT,\cB_\cT)$
\[ \mu_n=\sum_{t_j\in\pi_n}(X(t_{j+1})-X(t_j))^2\delta_{t_j}\]
(where $\delta_t$ denotes the Dirac measure at $t$) 
converges weakly to a measure $\mu$, possibly along some sub-sequence. The distribution function of $\mu$ is denoted $\langle X\rangle_t$, and we denote by $\cQ$ the set of continuous functions having the quadratic variation property.
\end{definition}
\begin{theorem}[\citep{Fol81}]\label{fol} For $X\in\cQ$ and  $f\in C^2$, the limit
$$\int_0^t f'(X_s)\id X_s:=\lim_{n\to \infty} \sum_{t_j\in \pi_n}f^\prime(X_{t_j})(X_{t_{j+1}}-X_{t_{j}})$$
is well defined and satisfies the pathwise It\^o formula:
\be f(X_t)-f(X_0)=\int_0^tf^\prime(X_s)\id X_s+\frac{1}{2}\int_0^tf^{\prime\prime}(X_s)\id \langle X\rangle_s.\label{ito}\ee
\end{theorem}
\proof \emph{(Sketch)}. The proof proceeds by writing the  expansion
\be
f(X_t)-f(X_0)=\sum_jf^\prime(X_{t_j})(X_{t_{j+1}}-X_{t_{j}})+\frac{1}{2}\sum_j
f^{\prime\prime}(X_{t_j})(X_{t_{j+1}}-X_{t_{j}})^2+\sum_jR(X_{t_j},X_{t_{j+1}})\label{taylor}\ee
where $R(a,b)\leq\phi(|b-a|)(b-a)^2$ with $\phi$ an increasing function, $\phi(c)\to0$ as $c\to 0$. As $m(\pi)\to0$, the third term on the right of \eqref{taylor} converges to 0 and the second term converges to the second term in \eqref{ito}. This shows that the first term converges, so that essentially the `pathwise stochastic integral' is \emph{defined by \eqref{ito}} for arbitrary $C^1$ functions $f^\prime$.\hfill$\square$

\medskip An important remark is that $X_t$, being continuous, achieves its  minimum and maximum $X_*, \,X^*$ in $[0,T]$, so only the values of $f$ in the compact interval $[X_*,X^*]$ (depending on the path $X$) are relevant in \eqref{ito}.

For the applications in Section \ref{sec-wvs} we need an `It\^o formula' valid when $f''$ is merely locally integrable, rather than continuous as it is in \eqref{ito}. In the theory of continuous semimartingales, such an extension proceeds via local time and the Tanaka formula---see Theorem VI.1.2 of \citet{rogwil}. Here we present a pathwise version, following the diploma thesis of \citet{dipl}. For a $C^2$ function $f$ we have
$ f(b)=f(a)+\int_a^bf^\prime(y)dy$,
and applying the same formula to $f'$ we obtain
\bstar f(b)-f(a)&=& \int_a^b\left(f\p(a)+\int_a^yf\pp(u)\id u\right)\id y 
= f\p(a)(b-a)+\int_a^b(b-u)f\pp(u)\id u \\
&=& f\p(a)(b-a)+\int_{-\infty}^\infty\1_{[a\wedge b,a\vee b]}(u)|b-u|f\pp(u)\id u.
\estar
Hence for any partition $\pi=\{t_j\}$ of $[0,t]$ we have the identity
\be f(X_t)-f(X_0)=\sum_jf\p(X_{t_j})(X_{t_{j+1}}-X_{t_j})+\int_{-\infty}^\infty\sum_j\left(1_{ [X_j^{\min},X_j^{\max}]}(u)|X_{t_{j+1}} -u|\right)f''(\id u),\label{id}\ee
where
$X_j^{\min}=\min\{X_{t_j},X_{t_{j+1}}\}$ and $X_j^{\max}=\max\{X_{t_j},X_{t_{j+1}}\}$.
We define
\be L^\pi_t(u)=2\sum_j1_{[ X_j^{\min},X_j^{\max}]}(u)|X_{t_{j+1}} -u|,\label{Lpi}\ee
and note that $L^\pi_t(u)=0$ for $u\notin(X_*,X^*)$.
\begin{definition}\label{L-def} Let $\cL$ be the set of continuous paths on $[0,T]$ such that the discrete pathwise local time $L^{\pi_n}_t(u)$ converges weakly in $\Ltwo(\id u)$ to a limit $L_t(\cdot)$, for each $t\in[0,T]$. 
\end{definition}
\begin{proposition}\label{LsubQ} $\cL\subset\cQ$. For $X\in\cL$ and $t\in[0,T]$ we have the occupation density formula
\be \int_AL_t(u)\id u=\int_0^t\1_A(X_s)\id\langle X\rangle_s,\qquad A\in\cB(\bR).\label{od}\ee
\end{proposition}
\proof Suppose $X\in \cL$, and let $\pi_n$ be a sequence of partitions of $[0,T]$ with $m(\pi_n)\to0$. From \eqref{id} we have, for $f\in C^2$ and $t\in[0,T]$,
\be f(X_{\tilde{t}_n})-f(X_0)-\sum_jf\p(X_{t_j})(X_{t_{j+1}}-X_{t_j})=\frac{1}{2}\int_{-\infty}^{\infty}
L^{\pi_n}_t(u)f\pp(u)\id u,\label{**}\ee
where $\tilde{t}_n$ is the nearest partition point in $\pi_n$ to $t$. As $n\to\infty$, the first term on the left of \eqref{**} converges to $f(X_t)$ and, since $f\pp\in \Ltwo([X_*,X^*],\id u)$, $X\in\cL$ implies that the right-hand side converges to $(1/2)\int L_t(u)f\pp(u)du$. Hence the `integral' term on the left of \eqref{**} also converges to, say, $I_t$. If we take $f(x)=x^2$ then the left-hand side of \eqref{**} is equal to 
\[\sum_{j:t_{j+1}\leq \tilde{t}_n}X^2_{t_{j+1}}-X^2_{t_{j}}-2X_{t_{j}}(X_{t_{j+1}}-X_{t_{j}})=\sum_{j:t_{j+1}\leq \tilde{t}_n} (X_{t_{j+1}}-X_{t_{j}})^2,\]
showing that (a) $X\in\cQ$ and (b) $I_t=\int_0^tf'(X_s)dX_s$, the F\"ollmer integral of Theorem \ref{fol}. It now follows from \eqref{**} that
\be\int_{-\infty}^{\infty}L_t(u)f\pp(u)\id u=\int_0^tf\pp(X_s)\id\langle X\rangle_s,\label{***}\ee
i.e., for any continuous function $g$ we have
\[\int_{-\infty}^{\infty}L_t(u)g(u)\id u=\int_0^tg(X_s)\id\langle X\rangle_s.\]
Approximating the indicator function $1_A$ by continuous functions and using the monotone convergence theorem, we obtain \eqref{od}.\hfill$\square$

\medskip For the next result, let $\cW_2$ be the set of functions $f$ in $C^1(\bR)$ such that $f'$ is weakly differentiable with derivative $f''$ in $\Ltwoloc(\bR)$.
\begin{theorem}\label{extended-ito} If $X\in \cL$ the F\"ollmer integral extends in such a way that the pathwise Ito formula \eqref{ito} is valid for $f\in\cW_2$.
\end{theorem}
\proof Let $\phi$ be a mollifier function, a non-negative $C^{\infty}$ function on $\bR$ such that $\phi(x)=0$ for $|x|>1$ and $\int\phi(x)dx=1$, define $\phi_n(x)=n\phi(nx)$ and let $f_n(x)=\int f(x-y)\phi_n(y)dy$. Then $f_n, f'_n$ converge pointwise to $f,f'$ and $f''_n\to f''$ weakly in $\Ltwo$. From Proposition \ref{LsubQ} we know that $X\in\cQ$ and as $f_n\in C^2$ we have, using \eqref{ito} and \eqref{***},
\[ f_n(X_t)-f_n(X_0)=\int_0^tf'_n(X_s)\id X_s+\frac{1}{2}\int_{-\infty}^{\infty}L_t(u)f_n''(u)\id u.\]
As $n\to\infty$, the left-hand side converges to $f(X_t)-f(X_0)$ and the second term on the right converges to $\frac{1}{2}\int_{-\infty}^{\infty}L_t(u)f''(u)du$, and we can define
\[\int_0^tf'_n(X_s)\id X_s=\lim_{n\to\infty}\left\{f_n(X_t)-f_n(X_0)
-\frac{1}{2}\int_{-\infty}^{\infty}L_t(u)f_n''(u)\id u\right\}.\]
It now follows from Proposition \ref{LsubQ} that the It\^o formula \eqref{ito} holds for this extended integral.
\hfill$\square$

\medskip We need one further result. 
\begin{proposition}\label{YfX} Let $X\in\cL$ and let $f:\bR\to\bR$ be a monotone $C^2$ function. Then $Y=f(X)\in\cL$ and the pathwise local times are related by
\be L^Y_t(u)=|f'(f^{-1}(u))|\,L^X_t(f^{-1}(u)).\label{Ylt}\ee
\end{proposition}
\proof We assume $f$ is increasing---the argument is the same if it is decreasing---and denote $v=f^{-1}(u)$, so $u=f(v)$. For the partition $\pi_n$ we have for fixed $t\in\cT$, from \eqref{Lpi},
\bea L^{Y,\pi_n}_t(u)&=&\sum_j\1_{\{f(X^{\mathrm{min}}_{n,j})\leq f(v)\leq f(X^{\mathrm{max}}_{n,j})\}}|f(X_{t^n_{j+1}})-f(v)|\nonumber\\
&=&\sum_j\1_{\{X^{\mathrm{min}}_{n,j})\leq v\leq X^{\mathrm{max}}_{n,j})\}}|f'(v)(X_{t^n_{j+1}}-v))
+\frac{1}{2}f''(\xi)(X_{t^n_{j+1}}-v)^2|\nonumber\\
&=&f'(v)\sum_j\1_{\{X^{\mathrm{min}}_{n,j})\leq v\leq X^{\mathrm{max}}_{n,j})\}}\left|X_{t^n_{j+1}}-v
+\frac{f''}{2f'}(X_{t^n_{j+1}}-v)^2\right|\label{****}
\eea
for $\xi$ between $v$ and $X_{t^n_{j+1}}$. Noting that $f''(\xi)/2f'(v)$ is bounded for $\xi,v\in[X_*,X^*]$ and that \\$\lim_{n\to\infty}\max_j|X(t^n_{j+1})-X(t^n_j)|=0$, we easily conclude that if $X\in\cL$ then the expression at \eqref{****} converges in $\Ltwo(\id v)$ to $f'(v)L^X_t(v)$, so that $Y\in\cL$ with local time given by \eqref{Ylt}.\hfill$\square$

\medskip We note that in the above result it suffices to assume that $f$ is defined on $[X_*,X^*]$. This allows us to apply the Proposition for $f=\log$ and and stock price trajectories in Section \ref{sec-wvs}. Indeed, from \eqref{Ylt} we obtain the elegant formula
\[ L^{\log X}_t(u)=e^{-u}L^X_t(e^u).\]

Finally, we build the connection between the pathwise calculus and the classical It\^o (stochastic) calculus.
\begin{theorem}\label{prob_conn} Let $(X_t, t\in[0,T])$ be a continuous semimartingale on some complete probability space $(\Omega,\cF,\bP)$. Then there is a set $N\in\cF$ such that $\bP N=0$ and for $\omega\notin N$ the path $t\mapsto X(t,\omega)$ belongs to $\cL$ (and hence to $\cQ$), and $L_t(u)$ defined in Definition \ref{L-def} coincides with the semimartingale local time of $X_t$ at $u$. 
\end{theorem}
\proof First, $X_t\in\cQ$ a.s. Indeed, Theorem IV.1.3 of \cite{revuz_yor} asserts that discrete approximations to the quadratic variation always converge in probability, so a sub-sequence converges almost surely, showing that $X_t\in\cQ$ in accordance with Definition \ref{def-qvp}. It is shown by \cite{Fol81} that the It\^o integral and the pathwise integral defined by \eqref{ito} coincide almost surely. Every semimartingale $S_t$ has an associated local time which satisfies the occupation density formula \eqref{od}. It remains to show that with probability 1 the discrete approximations $L^{\pi}_t$ defined by \eqref{Lpi} converge in $\Ltwo(\id u)$. This is proved in \cite{dipl}, by detailed estimates which we cannot include here.
\hfill$\square$

\bibliographystyle{elsart-harv}

\end{document}